\documentclass[letterpaper,11pt]{article}

\usepackage[T1]{fontenc}
\usepackage[total={8.5in,11in},top=1in,bottom=1in,left=1in,right=1in]{geometry}
\usepackage{amsmath,amsthm}
\usepackage{amssymb,latexsym}
\usepackage[mathscr]{eucal}
\usepackage{dsfont}
\usepackage{graphicx,color}
\usepackage{psfrag}
\usepackage{xspace}
\usepackage{comment}
\input epsf
\usepackage[unicode,final=true]{hyperref} 
\usepackage{soul}

\usepackage{fig/shellings/shelling-figs}

\newtheorem{theorem}{Theorem}

\newtheorem{corollary}[theorem]{Corollary}
\newtheorem{definition}[theorem]{Definition}

\newtheorem{lemma}[theorem]{Lemma}

\newcommand{\reals}{\mathbb{R}}
\newcommand{\euc}{\mathbb{E}}

\newcommand{\bB}{\mathcal{B}}
\newcommand{\cC}{\mathcal{C}}
\newcommand{\kK}{\mathcal{K}}
\newcommand{\fF}{\mathcal{F}}

\newcommand{\sS}{\mathcal{S}}
\newcommand{\yY}{\mathcal{Y}}
\newcommand{\SSS}{\mathscr{S}}

\renewcommand{\to}{\rightarrow}
\newcommand{\sub}{\leftarrow}

\newcommand{\sm}{\setminus}
\newcommand{\nin}{\not\in}

\newcommand{\mb}[1]{{\boldsymbol{#1}}}

\newcommand{\mc}{\mb{\gamma}}

\newcommand{\sep}{,\xspace}
\newcommand{\resp}{resp.,\xspace}
\newcommand{\ie}{i.e.,\xspace}
\newcommand{\eg}{e.g.,\xspace}
\newcommand{\lexp}[1]{\lfloor\frac{#1}{2}\rfloor}
\newcommand{\bx}[1]{\partial{#1}}
\newcommand{\link}[2][v]{{#2}/{#1}}
\newcommand{\linkfull}[2][v]{\text{link}(#1,#2)}
\newcommand{\str}[2][v]{\text{star}(#1,#2)}
\newcommand{\shell}[1]{\mathds{S}(#1)}
\newcommand{\bn}[2][V_1]{$(#2,#1)$-bineighborly\xspace}

\newcommand{\labelsize}{}

\begin{document}

\title{The maximum number of faces of the Minkowski sum\\ of two convex
  polytopes}

\author{Menelaos I. Karavelas$^{\dagger,\ddagger}$\hfil{}
Eleni Tzanaki$^{\dagger,\ddagger}$\\[5pt]
\it{}$^\dagger$Department of Applied Mathematics,\\
\it{}University of Crete\\
\it{}GR-714 09 Heraklion, Greece\\
{\small\texttt{\{mkaravel,etzanaki\}@tem.uoc.gr}}\\[5pt]
\it{}$^\ddagger$Institute of Applied and Computational Mathematics,\\
\it{}Foundation for Research and Technology - Hellas,\\
\it{}P.O. Box 1385, GR-711 10 Heraklion, Greece}



\phantomsection
\addcontentsline{toc}{section}{Title}

\maketitle

\phantomsection
\addcontentsline{toc}{section}{Abstract}

\begin{abstract}
  We derive tight expressions for the maximum
number of $k$-faces, $0\le{}k\le{}d-1$, of the
Minkowski sum, $P_1\oplus{}P_2$, of two $d$-dimensional convex polytopes
$P_1$ and $P_2$, as a function of the number of vertices of the polytopes.

For even dimensions $d\ge{}2$, the maximum values are attained when
$P_1$ and $P_2$ are cyclic $d$-polytopes with disjoint vertex sets.
For odd dimensions $d\ge{}3$, the maximum values are attained when
$P_1$ and $P_2$ are $\lfloor\frac{d}{2}\rfloor$-neighborly
$d$-polytopes, whose vertex sets are chosen appropriately from two
distinct $d$-dimensional moment-like curves.

  \bigskip\noindent
  \textit{Key\;words:}\/
  high-dimensional geometry\sep discrete geometry\sep
  combinatorial geometry\sep combinatorial complexity\sep
  Minkowski sum\sep convex polytopes
  \medskip\par\noindent
  \textit{2010 MSC:}\/ 52B05, 52B11, 52C45, 68U05
\end{abstract}

\section{Introduction}
\label{sec:intro}

Given two $d$-dimensional polytopes, or simply $d$-polytopes, $P$ and
$Q$, their Minkowski sum, $P\oplus{}Q$, is defined as the set
$\{p+q\,\,|\,\,p\in{}P, q\in{}Q\}$.
Minkowski sums are fundamental structures in both Mathematics and
Computer Science. They appear in a variety of different subjects,
including Combinatorial Geometry, Computational Geometry, Computer
Algebra, Computer-Aided Design \& Solid Modeling, Motion Planning,
Assembly Planning, Robotics (see \cite{w-mspcc-07,f-mscaa-08} and the
references therein), and, more recently,
Game Theory \cite{r-gtsis-00}, Computational Biology \cite{ps-ascb-05}
and Operations Research \cite{z-pomdp-10}.

Despite their apparent importance, little is known about the
worst-case complexity of Minkowski sums in dimensions four and
higher. In two dimensions, the worst-case complexity of Minkowski sums
is well understood. Given two convex polygons $P$ and
$Q$ with $n$ and $m$ vertices, respectively, the maximum number
of vertices and edges of $P\oplus{}Q$ is $n+m$
\cite{bkos-cgaa-00}. This result can be immediately generalized (\eg
by induction) to any number of summands. If $P$ is convex and $Q$ is
non-convex (or vice versa), the worst-case complexity of $P\oplus{}Q$
is $\Theta(nm)$, while if both $P$ and $Q$ are non-convex the complexity
of their Minkowski sum can be as high as $\Theta(n^2m^2)$
\cite{bkos-cgaa-00}. When $P$ and $Q$ are 3-polytopes
(embedded in the 3-dimensional Euclidean space), the worst-case
complexity of $P\oplus{}Q$ is $\Theta(nm)$, if both $P$ and $Q$ are
convex, and $\Theta(n^3m^3)$, if both $P$ and $Q$ are non-convex
(\eg see \cite{fhw-emcms-09}). For the intermediate cases, \ie if only
one of $P$ and $Q$ is convex, see \cite{s-amp-04}.

Given two convex $d$-polytopes $P_1$ and $P_2$ in $\euc^d$,
$d\ge{}2$, with $n_1$ and $n_2$ vertices, respectively, we can easily get
a straightforward upper bound of $O((n_1+n_2)^{\lexp{d+1}})$ on the
complexity of $P_1\oplus{}P_2$ by means of the following reduction:
embed $P_1$ and $P_2$ in the hyperplanes $\{x_{d+1}=0\}$ and
$\{x_{d+1}=1\}$ of $\euc^{d+1}$, respectively; then the weighted
Minkowski sum
$(1-\lambda)P_1\oplus\lambda{}P_2=
\{(1-\lambda)p_1+\lambda{}p_2\,\,|\,\,p_1\in{}P_1,p_2\in{}P_2\}$,
$\lambda\in(0,1)$, of $P_1$ and $P_2$ is the intersection of the
convex hull, $CH_{d+1}(\{P_1,P_2\})$, of $P_1$ and $P_2$ with the hyperplane
$\{x_{d+1}=\lambda\}$.
The embedding and reduction described above are essentially
what are known as the \emph{Cayley embedding} and \emph{Cayley trick},
respectively \cite{hrs-ctlsb-00}.
From this reduction it is obvious that the worst-case complexity of
$(1-\lambda)P_1\oplus\lambda{}P_2$ is bounded from above by the complexity
of $CH_{d+1}(\{P_1,P_2\})$, which is $O((n_1+n_2)^{\lexp{d+1}})$.
Furthermore, the complexity of the weighted Minkowski sum of $P_1$ and
$P_2$ is independent of $\lambda$, in the sense that for any value of
$\lambda\in(0,1)$ the polytopes we get by intersecting
$CH_{d+1}(\{P_1,P_2\})$ with $\{x_{d+1}=\lambda\}$ are combinatorially
equivalent. In fact, since $P_1\oplus{}P_2$ is nothing but
$\frac{1}{2}P_1\oplus\frac{1}{2}P_2$ scaled by a factor of 2, the
complexity of the weighted Minkowski sum of two convex polytopes is
the same as the complexity of their unweighted Minkowski sum. 
Very recently (cf. \cite{kt-chsch-11b}), the authors of this paper have
considered the problem of computing the asymptotic worst-case
complexity of the convex hull of a fixed number $r$ of convex
$d$-polytopes lying on $r$ parallel hyperplanes of $\euc^{d+1}$.
A direct corollary of our results is a tight bound on the
worst-case complexity of the Minkowski sum of two convex $d$-polytopes
for all odd dimensions $d\ge{}3$, which refines the ``obvious'' upper
bound mentioned above. More precisely, we have shown that for
$d\ge{}3$ odd, the worst-case complexity of $P_1\oplus{}P_2$ is in
$\Theta(n_{1}n_2^{\lexp{d}}+n_{2}n_1^{\lexp{d}})$, which is a refinement of the
obvious upper bound when $n$ and $m$ asymptotically differ. 

In terms of exact bounds on the number of faces of the Minkowski sum
of two polytopes, results are known only when the two summands are
convex. Besides the trivial bound for convex polygons (2-polytopes),
mentioned in the previous paragraph, the first result of this nature
was shown by Gritzmann and Sturmfels \cite{gs-mapcc-93}: given $r$
polytopes $P_1,P_2,\ldots,P_r$ in $\euc^d$, with a total of $n$
non-parallel edges, the number of $l$-faces,
$f_l(P_1\oplus{}P_2\oplus{}\cdots\oplus{}P_r)$,
of $P_1\oplus{}P_2\oplus{}\cdots\oplus{}P_r$ is bounded from above
by $2\binom{n}{l}\sum_{j=0}^{d-1-l}\binom{n-l-1}{j}$. This bound is
attained when the polytopes $P_i$ are \emph{zonotopes}, and their
generating edges are in general position.

Regarding bounds as a function of the number of vertices
or facets of the summands, Fukuda and Weibel \cite{fw-fmacp-07} have
shown that, given two 3-polytopes $P_1$ and $P_2$ in $\euc^3$, the
number of $k$-faces of $P_1\oplus{}P_2$, $0\le{}k\le{}2$, is bounded
from above as follows:
\begin{equation}\label{equ:ub3}
  \begin{aligned}
    f_0(P_1\oplus{}P_2)&\le{}n_1 n_2,\\
    f_1(P_1\oplus{}P_2)&\le{}2n_1 n_2+n_1+n_2-8,\\
    f_2(P_1\oplus{}P_2)&\le{}n_1 n_2+n_1+n_2-6.
  \end{aligned}
\end{equation}
where $n_j$ is the number of vertices of $P_j$, $j=1,2$.
Weibel \cite{w-mspcc-07} has also derived similar expressions in terms
of the number of facets $m_j$ of $P_j$, $j=1,2$, namely:
\begin{align*}
  f_0(P_1\oplus{}P_2)&\le{}4m_1 m_2-8m_1-8m_2+16,\\
  f_1(P_1\oplus{}P_2)&\le{}8m_1 m_2-17m_1-17m_2+40,\\
  f_2(P_1\oplus{}P_2)&\le{}4m_1 m_2-9m_1-9m_2+26.
\end{align*}
All these bounds are tight.
Fogel, Halperin and Weibel \cite{fhw-emcms-09} have further
generalized some of these bounds in the case of $r$ summands. More
precisely, they have shown that given $r$ 3-polytopes
$P_1,P_2,\ldots,P_r$ in $\euc^3$, where $P_j$ has $m_j\ge{}d+1$ facets, the
number of facets of the Minkowski sum
$P_1\oplus{}P_2\oplus{}\cdots\oplus{}P_r$ is bounded from above by
\[\sum_{1\le{}i<j\le{}r}(2m_i-5)(2m_j-5)+\sum_{i=1}^r{}m_i+\binom{r}{2},\]
and this bound is tight.

For dimensions four and higher there are no results that relate the
worst-case number of $k$-faces of the Minkowski sum of two of more
convex polytopes with the number of facets of the summands. There are,
however, bounds on the number of $k$-faces of the Minkowski sum of
convex polytopes, as a function of the number of vertices of the
summands. More precisely, Fukuda and Weibel \cite{fw-fmacp-07} have
shown that the number of vertices of the Minkowski sum of $r$
$d$-polytopes $P_1,\ldots,P_r$, where $r\le{}d-1$ and $d\ge{}2$, is
bounded from above by $\prod_{i=1}^{r}n_i$, where $n_i$ is the number
of vertices of $P_i$, and this bound is tight. On the other hand, for
$r\ge{}d$ this bound cannot be attained: Sanyal \cite{s-tovnm-09} has
shown that for $r\ge{}d$, $f_0(P_1\oplus\cdots\oplus{}P_r)$ is bounded
from above by
$\left(1-\frac{1}{(d+1)^d}\right)\prod_{i=1}^{r}n_i$, which is,
clearly, strictly smaller than $\prod_{i=1}^{r}n_i$.
For higher-dimensional faces, \ie for $k\ge{}1$, Fukuda and Weibel
\cite{fw-fmacp-07} have shown that the number of $k$-faces of the
Minkowski sum of $r$ polytopes is bounded as follows:
\begin{equation}\label{equ:ub-trivial}
  f_k(P_1\oplus{}P_2\oplus\cdots\oplus{}P_r)\le
  \sum_{\substack{1\le{}s_i\le{}n_i\\s_1+\ldots+s_r=k+r}}\prod_{i=1}^r\binom{n_i}{s_i},
  \qquad 0\le{}k\le{}d-1,
\end{equation}
where $n_i$ is the number of vertices of $P_i$. These bounds are
tight for $d\ge{}4$, $r\le\lexp{d}$, and for all $k$ with
$0\le{}k\le\lexp{d}-r$, i.e., for the cases where both the number of
summands and the dimension of the faces considered is
small. 

We end our discussion of the previous work related to this paper by
some results presented in a technical report of Weibel \cite{w-mfmsl-10}.
In this report, Weibel considers the case where the number of summands, $r$,
is at least as big as the dimension of the polytopes. In this setting
he gives a relation between the number of $k$-faces of the Minkowski
sum of $r$ polytopes, $r\ge{}d\ge{}2$, and the number of $k$-faces of
the Minkowski sum of subsets of the original set of $r$ polytopes, that
are of size at most $d-1$. In more detail, if we have $r$
$d$-polytopes $P_1,P_2,\ldots,P_r$ in $\euc^d$, where $r\ge{}d$, that are
in general position, then the following relation holds for any $k$
with $0\le{}k\le{}d-1$:
\begin{equation}\label{equ:many-summands}
  f_k(P_1\oplus{}P_2\oplus\cdots{}\oplus{}P_r)
  =\alpha+\sum_{j=1}^{d-1}(-1)^{d-1-j}\binom{r-1-j}{d-1-j}
  \sum_{S\in\cC_j^r}(f_k(P_S)-\alpha)
  \le\sum_{S\in\cC_{d-1}^r}f_k(P_S),
\end{equation}
where $\cC_j^r$ is the family of subsets of $\{1,2,\ldots,r\}$ of
cardinality $j$, $P_S$ is the Minkowski sum of the polytopes in $S$,
and, finally, $\alpha=2$ if $k=0$ and $d$ is odd, $\alpha=0$,
otherwise. Weibel then uses this relation to derive upper bounds on the
number of vertices of the Minkowski sum of $r$ $d$-polytopes in
$\euc^d$, when $r\ge{}d$. An important qualitative consequence of
relation \eqref{equ:many-summands} is that, when we consider Minkowski
sums of $d$-polytopes in $\euc^d$, essentially it really suffices to
consider up to $d-1$ summands. To pose it otherwise, if we know
good/tight upper bounds for the worst-case number of $k$-faces of the
Minkowski sum of $d-1$ polytopes, then we immediately know upper
bounds for the $k$-faces of the Minkowski sum of $r\ge{}d$ polytopes
in $\euc^d$. If we are to strive for tight exact bounds, however,
there is still something to be done in this case, due to the fact that
the sum in \eqref{equ:many-summands} is an alternating sum: not only
do we have to find sets of $r\ge{}d$ polytopes such that any subset of
them of size at most $d-1$ yields the worst possible number of faces,
but also prove that such a configuration does indeed maximize the
right-hand size of the equality in \eqref{equ:many-summands}.

In this paper, we extend previous results on the exact maximum number
of faces of the Minkowski sum of two convex $d$-polytopes\footnote{In
  the rest of the paper, all polytopes are considered to be convex.}.
More precisely, we show that given two $d$-polytopes $P_1$ and
$P_2$ in $\euc^d$ with $n_1\ge{}d+1$ and $n_2\ge{}d+1$ vertices,
respectively, the maximum number of $k$-faces of the Minkowski sum
$P_1\oplus{}P_2$ is bounded from above as follows:
\begin{equation*}
  f_{k-1}(P_1\oplus{}P_2)\le{}f_k(C_{d+1}(n_1+n_2))
  -\sum_{i=0}^{\lexp{d+1}}\binom{d+1-i}{k+1-i}
  \left(\binom{n_1-d-2+i}{i}+\binom{n_2-d-2+i}{i}\right),
\end{equation*}
where $1\le{}k\le{}d$, and $C_d(n)$ stands for the cyclic $d$-polytope
with $n$ vertices. The expressions above are shown to be tight for any
$d\ge{}2$ and for all $1\le{}k\le{}d$, and, clearly, match
with the corresponding expressions for two and three dimensions
(cf. rel. \eqref{equ:ub3}), as well as the expressions in
\eqref{equ:ub-trivial} for $r=2$ and for all
$0\le{}k\le{}\lexp{d}-2$.

To prove the upper bounds we use the embedding in one dimension higher
already stated above. Given the $d$-polytopes $P_1$ and
$P_2$ in $\euc^d$, we embed $P_1$ and $P_2$ in the hyperplanes
$\{x_{d+1}=0\}$ and $\{x_{d+1}=1\}$ of $\euc^{d+1}$. We consider the
convex hull $P=CH_{d+1}(\{P_1,P_2\})$ and argue that, for the purposes
of the worst-case upper bounds, it suffices to consider the case where
$P$ is simplicial, except possibly for its two facets $P_1$ and
$P_2$. We concentrate on the set $\fF$ of faces of $P$ that are neither
faces of $P_1$ nor faces of $P_2$. The reason that we focus on
$\fF$ is that there is a bijection between the $k$-faces of
$\fF$ and the $(k-1)$-faces of $P_1\oplus{}P_2$, $1\le{}k\le{}d$, and, thus,
deriving upper bounds of the number of $(k-1)$-faces of
$P_1\oplus{}P_2$ reduces to deriving upper bounds for the number of
$k$-faces of $\fF$. We then proceed in a manner analogous to that used
by McMullen \cite{m-mnfcp-70} to prove the Upper Bound Theorem for
polytopes.
We consider the $f$-vector $\mb{f}(\fF)$ of $\fF$, from this we
define the $h$-vector $\mb{h}(\fF)$ of $\fF$, and continue by:
\begin{enumerate}
\item deriving Dehn-Sommerville-like equations for $\fF$, expressed in terms of
  the elements of $\mb{h}(\fF)$ and the $g$-vectors of the
  boundary complexes of $P_1$ and $P_2$, and,
\item establishing a recurrence relation for the elements of
  $\mb{h}(\fF)$.
\end{enumerate}
From the latter, we inductively compute upper bounds on the elements of
$\mb{h}(\fF)$, which we combine with the Dehn-Sommerville-like
equations for $\fF$, to get refined upper bounds for the ``left-most
half'' of the elements of $\mb{h}(\fF)$, \ie for the values
$h_k(\fF)$ with $k>\lexp{d+1}$.
We then establish our upper bounds by computing $\mb{f}(\fF)$ from
$\mb{h}(\fF)$.

To prove the lower bounds we distinguish between even and odd
dimensions. In even dimensions $d\ge{}2$, we show that the $k$-faces
of the Minkowski sum of any two cyclic $d$-polytopes with
$n_1$ and $n_2$ vertices, respectively, whose vertex sets are
distinct, attain the upper bounds we have proved. In odd dimensions $d\ge{}3$,
the construction that establishes the tightness of our bounds is more
intricate. We consider the $(d-1)$-dimensional moment curve
$\mc(t)=(t,t^2,t^3,\ldots,t^{d-1})$, $t>0$, and define two vertex
sets $V_1$ and $V_2$ with $n_1$ and $n_2$ vertices on $\mc(t)$,
respectively. We then embed $V_1$ (\resp $V_2)$ on the hyperplane
$\{x_2=0\}$ (\resp $\{x_1=0\}$) of $\euc^d$ and perturb the
$x_2$-coordinates (\resp $x_1$-coordinates) of the vertices in $V_1$
(\resp $V_2$), so that the polytope $P_1$ (\resp $P_2$) defined as the
convex hull, in $\euc^d$, of the vertices in $V_1$ (\resp $V_2$) is
full-dimensional.
We then argue that by \emph{appropriately choosing} the vertex sets
$V_1$ and $V_2$, the number of $k$-faces of the Minkowski sum
$P_1\oplus{}P_2$ attains its maximum possible value. At a very
high/qualitative level, the appropriate choice we refer to above
amounts to choosing $V_1$ and $V_2$ so that the parameter values on 
$\mc(t)$ of the vertices in $V_1$ and $V_2$, lie within two disjoint
intervals of $\reals$ that are far away from each other. 

The structure of the rest of the paper is as follows. In Section
\ref{sec:prelim} we formally give various definitions, and recall a
version of the Upper Bound Theorem for polytopes that will be useful
later in the paper. In Section \ref{sec:bn} we define what we call
\emph{bineighborly} polytopal complexes and prove some properties
associated with them. The reason that we introduce this new notion is
the fact that the tightness of our upper bounds is shown to be
equivalent to requiring that the $(d+1)$-polytope
$P=CH_{d+1}(\{P_1,P_2\})$, defined above, is bineighborly. In Section
\ref{sec:ub} we prove our upper bounds on
the number of faces of the Minkowski sum of two polytopes. In Section
\ref{sec:lb} we describe our lower bound constructions and show that
these constructions attain the upper bounds proved in Section
\ref{sec:ub}. We conclude the paper with Section \ref{sec:concl},
where we summarize our results, and state open problems and directions
for future work.
 
\section{Definitions and preliminaries}
\label{sec:prelim}

A \emph{convex polytope}, or simply \emph{polytope}, $P$ in
$\euc^d$ is the convex hull of a finite set of points $V$ in
$\euc^d$, called the \emph{vertex set} of $P$.
A polytope $P$ can equivalently be described as the intersection of
all the closed halfspaces containing $V$. 
A \emph{face} of $P$ is the intersection of $P$ with a hyperplane
for which the polytope is contained in one of the two closed
halfspaces delimited by the hyperplane.
The dimension of a face of $P$ is the dimension of its affine hull.
A $k$-face of $P$ is a $k$-dimensional face of $P$.
We consider the polytope itself as a trivial $d$-dimensional face; all
the other faces are called \emph{proper} faces. We use the term
\emph{$d$-polytope} to refer to a polytope the trivial face of which
is $d$-dimensional.
For a $d$-polytope $P$, the $0$-faces of $P$ are its
\emph{vertices}, the $1$-faces of $P$ are its \emph{edges}, the
$(d-2)$-faces of $P$ are called \emph{ridges}, while the
$(d-1)$-faces are called \emph{facets}.
For $0\le{}k{}\le{}d$ we denote by $f_k(P)$ the number of $k$-faces
of $P$.
Note that every $k$-face $F$ of $P$ is also a $k$-polytope whose
faces are all the faces of $P$ contained in $F$.
A $k$-simplex in $\euc^d$, $k\le{}d$, is the convex hull of any
$k+1$ affinely independent points in $\euc^d$.
A polytope is called \emph{simplicial} if all its proper faces are
simplices. Equivalently, $P$ is simplicial if for every vertex $v$
of $P$ and every face $F\in{}P$, $v$ does not belong to the affine
hull of the vertices in $F\sm\{v\}$.


%
A \emph{polytopal complex} $\cC$ is a finite collection of polytopes
in $\euc^d$ such that
(i) $\emptyset\in\cC$,
(ii) if $P\in\cC$ then all the faces of $P$ are also in $\cC$ and
(iii) the intersection $P\cap{}Q$ for two polytopes $P$ and $Q$ in
$\cC$ is a face of both $P$ and $Q$. 
The dimension $\dim(\cC)$ of $\cC$ is the largest dimension of a
polytope in $\cC$.
A polytopal complex is called \emph{pure} if all its maximal (with
respect to inclusion) faces have the same dimension. 
In this case the maximal faces are called the \emph{facets} of
$\cC$. We use the term \emph{$d$-complex} to refer to a
polytopal complex whose maximal faces are $d$-dimensional (\ie the
dimension of $\cC$ is $d$).
A polytopal complex is simplicial if all its faces are simplices. 
Finally, a polytopal complex $\cC'$ is called a \emph{subcomplex}
of a polytopal complex $\cC$ if all faces of $\cC'$ are also faces
of $\cC$.

One important class of polytopal complexes arise from polytopes.
More precisely, a $d$-polytope $P$, together with all its
faces and the empty set, form a $d$-complex, denoted
by $\cC(P)$. The only maximal face of $\cC(P)$, which is clearly
the only facet of $\cC(P)$, is the polytope $P$ itself.
Moreover, all proper faces of $P$ form a pure $(d-1)$-complex,
called the \emph{boundary complex} $\cC(\bx{}P)$, or simply
$\bx{}P$ of $P$. The facets of $\bx{}P$ are just the facets of $P$,
and its dimension is, clearly, $\dim(\bx{}P)=\dim(P)-1=d-1$.

Given a $d$-polytope $P$ in $\euc^d$, consider $F$ a facet of $P$,
and call $H$ the supporting hyperplane of $F$ (with respect to $P$).
For an arbitrary point $p$ in $\euc^d$, we say that $p$ is
\emph{beyond} (\resp \emph{beneath}) the facet $F$ of $P$, if $p$
lies in the open halfspace of $H$ that does not contain $P$ (\resp
contains the interior of $P$). Furthermore, we say that an arbitrary
point $v'$ is \emph{beyond} the vertex $v$ of $P$ if for every facet
$F$ of $P$ containing $v$, $v'$ is beyond $F$, while for every facet
$F$ of $P$ not containing $v$, $v'$ is beneath $F$.
For a vertex $v$ of $P$, the \emph{star} of $v$, denoted by $\str{P}$,
is the polytopal complex of all faces of $P$ that contain $v$, and
their faces. The \emph{link} of $v$, denoted by $\linkfull{P}$, is the
subcomplex of $\str{P}$ consisting of all the
faces of $\str{P}$ that do not contain $v$.

\begin{definition}[{\cite[Remark 8.3]{z-lp-95}}]
  \label{shellability}
  Let $\cC$ be a pure simplicial polytopal $d$-complex. A
  shelling $\shell{\cC}$ of $\cC$ is a linear ordering
  $F_1,F_2,\ldots,F_s$
  of the facets of $\cC$ such that for all $1<j\le{} s$ the
  intersection, $F_j\cap\left(\bigcup_{i=1}^{j-1}F_i\right)$, of the
  facet $F_j$ with the previous facets is non-empty and pure
  $(d-1)$-dimensional.\\
  In other words, for every $i<j$ there exists some $\ell <j$ such that
  the intersection $F_i\cap{}F_j$ is contained in $F_\ell\cap{}F_j$, 
  and such that $F_\ell\cap{}F_j $ is a facet of $F_j$. 
\end{definition}

Every polytopal complex that has a shelling is called \emph{shellable}.
In particular, the boundary complex of a polytope of always shellable.
(cf. \cite{bm-sdcs-71}).
Consider a pure shellable simplicial polytopal complex $\cC$ and let
$\shell{\cC}=\{F_1,\ldots,F_s\}$ be a shelling order of its facets.
The \emph{restriction} $R(F_j)$ of a facet $F_j$ is the set
of all vertices $v\in{}F_j$ such that $F_j\sm\{v\}$ is contained in
one of the earlier facets.\footnote{For simplicial faces, we
  identify the face with its defining vertex set.}
The main observation here is that when we construct $\cC$ according to
the shelling $\shell{\cC}$, the new faces at the $j$-th step of the
shelling are exactly the vertex sets $G$ with
$R(F_j)\subseteq{}G\subseteq{}F_j$ (cf. \cite[Section 8.3]{z-lp-95}).
Moreover, notice that $R(F_1)=\emptyset$ and $R(F_i)\ne{}R(F_j)$ for
all $i\ne{}j$.

The $f$-vector $\mb{f}(P)=(f_{-1}(P),f_0(P),\ldots,f_{d-1}(P))$
of a $d$-polytope $P$ (or its boundary complex $\bx{}P$) is
defined as the $(d+1)$-dimensional vector consisting of the number
$f_k(P)$ of $k$-faces of $P$, $-1\le{}k\le{}d-1$, where $f_{-1}(P)=1$
refers to the empty set.
The $h$-vector $\mb{h}(P)=(h_0(P),h_1(P),\ldots,h_d(P))$ of a
$d$-polytope $P$ (or its boundary complex $\bx{}P$) is
defined as the $(d+1)$-dimensional vector, where
$h_k(P):=\sum_{i=0}^{k}(-1)^{k-i}\binom{d-i}{d-k}f_{i-1}(P)$,
$0\le{}k\le{}d$. It is easy to verify from the
defining equations of the $h_k(P)$'s that the elements of
$\mb{f}(P)$ determine the elements of $\mb{h}(P)$ and
vice versa. 

For simplicial polytopes, the number $h_k(P)$ counts the number of
facets of $P$ in a shelling of $\bx{}P$, whose restriction has size
$k$; this number is independent of the particular shelling chosen
(cf. \cite[Theorem 8.19]{z-lp-95}). 
Moreover, the elements of $\mb{f}(P)$ (or,
equivalently, $\mb{h}(P)$) are not linearly independent; they
satisfy the so called \emph{Dehn-Sommerville equations}, which can be
written in a very concise form as:
$h_k(P)=h_{d-k}(P)$, $0\le{}k\le{}d$.
An important implication of the existence of the Dehn-Sommerville
equations is that if we know the face numbers $f_k(P)$ for all
$0\le{}k\le\lexp{d}-1$, we can determine the remaining face
numbers $f_k(P)$ for all $\lexp{d}\le{}k\le{}d-1$.
Both the $f$-vector and $h$-vector of a simplicial $d$-polytope are
related to the so called $g$-vector. For a simplicial $d$-polytope $P$
its $g$-vector is the $(\lexp{d}+1)$-dimensional vector
$\mb{g}(P)=(g_0(P),g_1(P),\ldots,$ $g_{\lexp{d}}(P))$, where $g_0(P)=1$, and
$g_k(P)=h_k(P)-h_{k-1}(P)$, $1\le{}k\le{}\lexp{d}$ (see also
\cite[Section 8.6]{z-lp-95}). Using the convention that
$h_{d+1}(P)=0$, we can actually extend the definition of $g_k(P)$ for
all $0\le{}k\le{}d+1$, while using the Dehn-Sommerville equations for
$P$ yields: $g_{d+1-k}(P)=-g_k(P)$, $0\le{}k\le{}d+1$. We can then
express $\mb{f}(P)$ in terms of $\mb{g}(P)$ as follows:
\begin{equation*}
  f_{k-1}(P)=\sum_{j=0}^{\lexp{d}}g_j(P)
  \left(\binom{d+1-j}{d+1-k}-\binom{j}{d+1-k}\right),\quad 0\le{}k\le{}d+1.
\end{equation*}
As a final note for this section, the Upper Bound Theorem for
polytopes can be expressed in terms of their $g$-vector:

\begin{corollary}[{\cite[Corollary 8.38]{z-lp-95}}]\label{cor:UBT}
  We consider simplicial $d$-polytopes $P$ of fixed dimension $d$ and
  fixed number of vertices $n=g_1(P)+d+1$.
  $\mb{f}(P)$ has its componentwise maximum if and only if all the
  components of $\mb{g}(P)$ are maximal, with 
  \begin{equation}
    g_k(P)=\binom{g_1(P)+k-1}{k}=\binom{n-d-2+k}{k}.
  \end{equation}
  Also, $f_{k-1}(P)$ is maximal if an only if $g_i(P)$ is maximal for
  all $i$ with $i\le\min\{k,\lexp{d}\}$. 
\end{corollary}


\section{Bineighborly polytopal complexes}
\label{sec:bn}

Let $\cC$ be a $d$-complex, and let $V$ be the vertex set of $\cC$.
Let $\{V_1,V_2\}$ be a partition of $V$ and define $\cC_1$ (\resp
$\cC_2$) to be the subcomplex of $\cC$ consisting of all the faces of
$\cC$ whose vertices are vertices in $V_1$ (\resp $V_2$). We start
with a useful definition:

\begin{definition}\label{def:bn}
  Let $\cC$ be a $d$-complex. We say that $\cC$ is \bn{k} if we can
  partition the vertex set $V$ of $\cC$ into two non-empty subsets
  $V_1$ and $V_2=V\sm{}V_1$ such that for every
  $\emptyset\subset{}S_j\subseteq{}V_j$, $j=1,2$, with
  $|S_1|+|S_2|\le{}k$, the vertices of $S_1\cup{}S_2$ define a face of
  $\cC$ (of dimension $|S_1|+|S_2|-1$).
\end{definition}

\noindent
We introduce the notion of bineighborly polytopal complexes because
they play an important role when considering the maximum
complexity of the Minkowski sum of two $d$-polytopes $P_1$ and
$P_2$. As we will see in the upcoming section, the number of
$(k-1)$-faces of $P_1\oplus{}P_2$ is maximal for all
$1\le{}k\le{}l$, $l\le\lexp{d-1}$, if and only if the convex hull $P$ of
$P_1$ and $P_2$, when embedded in the hyperplanes $\{x_{d+1}=0\}$ and
$\{x_{d+1}=1\}$ of $\euc^{d+1}$, respectively, is \bn{l+1},
where $V_1$ stands for the vertex set of $P_1$. Even more
interestingly, in any odd dimension $d\ge{}3$, the number of
$k$-faces of $P_1\oplus{}P_2$ is maximized for all
$0\le{}k\le{}d-1$, if and only if $P$ is \bn{\lexp{d+1}}.
In the rest of this section we highlight some properties of
bineighborly polytopal complexes that will be useful in the upcoming
sections.

A direct consequence of our definition is the following: suppose that
$\cC$ is a \bn{l} polytopal complex, and let $F$ be a $k$-face $F$ of
$\cC$, $1\le{}k<l$, such that at least one vertex of $F$ is in $V_1$
and at least one vertex of $F$ is in $V_2$; then $F$ is simplicial
(\ie $F$ is a $k$-simplex).
Another immediate consequence of Definition \ref{def:bn} is that a
$k$-neighborly $d$-complex is also \bn[V']{k} for every non-empty subset
$V'$ of its vertex set:

\begin{corollary}
  Let $\cC$ be a $k$-neighborly $d$-complex, with vertex set $V$. 
  Then, for every $V'$, with $\emptyset\subset{}V'\subset{}V$, $\cC$
  is \bn[V']{k}.
\end{corollary}

It is easy to see that if a $d$-complex $\cC$ is \bn{k}, then
it is $(k-1)$-neighborly, as the following straightforward lemma
suggests.

\begin{lemma}\label{lem:bn2n}
  Let $\cC$ be a \bn{k} $d$-complex, $k\ge{}2$.
  Then $\cC$ is $(k-1)$-neighborly.
\end{lemma}

\begin{proof}
  Let $S$ be a non-empty subset of $V$ of size $k-1$. Consider
  the following, mutually exclusive cases:
  \begin{enumerate}
  \item $S$ consists of vertices of both $V_1$ and $V_2$. In this case
    choose a vertex $v\in{}V\sm(V_1\cup{}V_2)$.
  \item $S$ consists of vertices of $V_1$ only. In this case choose
    a vertex $v\in{}V_2$.
  \item $S$ consists of vertices of $V_2$ only. In this case choose
    a vertex $v\in{}V_1$.
  \end{enumerate}
  Consider the vertex set $S'=S\cup\{v\}$, where $v$ is defined as
  above. $S'$ has size $k$, and has at least one vertex from $V_1$ and
  at least one vertex from $V_2$. Since $\cC$ is \bn{k}, the vertex
  set $S'$ defines a $(k-1)$-face $F_{S'}$ of $\cC$, which is, in
  fact, a $(k-1)$-simplex. This implies that $S$ is a $(k-2)$-face of
  $F_{S'}$, and thus a $(k-2)$-face of $\cC$. In other words, for
  every vertex subset $S$ of $\cC$ of size $k-1$, $S$ defines a
  $(k-2)$-face of $\cC$, \ie $\cC$ is $(k-1)$-neighborly.
\end{proof}

The following lemma is in some sense the reverse of Lemma
\ref{lem:bn2n}.

\begin{lemma}\label{lem:bn+n2n}
  Let $\cC$ be a \bn{k} $d$-complex, and let its two subcomplexes
  $\cC_1$ and $\cC_2$ be $k$-neighborly.
  Then $\cC$ is also $k$-neighborly.
\end{lemma}

\begin{proof}
  Let $S$ be a non-empty subset of $V$ of size $k$. Consider
  the following, mutually exclusive cases:
  \begin{enumerate}
  \item $S$ consists of vertices of both $V_1$ and $V_2$. Then, since
    $\cC$ is \bn{k}, $S$ defines a $(k-1)$-face of $\cC$.
  \item $S$ consists of vertices of $V_j$ only, $j=1,2$. Since $\cC_j$
    is $k$-neighborly, $S$ defines a $(k-1)$-face of $\cC_j$. However,
    $\cC_j$ is a subcomplex of $\cC$, which further implies that $S$
    is also a face of $\cC$.
  \end{enumerate}
  Hence, for every vertex subset $S$ of $V$ of size $k$, $S$ defines a
  $(k-1)$-face of $\cC$, \ie $\cC$ is $k$-neighborly.
\end{proof}

Consider again a $d$-complex $\cC$ with vertex set $V$. As
above, partition $V$ into two subsets $V_1$ and $V_2$, and let $\cC_1$
and $\cC_2$ be the corresponding subcomplexes of $\cC$. Finally,
let $\bB$ be the set of faces of $\cC$ that are not faces of either
$\cC_1$ or $\cC_2$.
We end this section with the following lemma that gives tight upper
bounds for the number of faces in $\bB$. In what follows, we denote by
$n_j$ the cardinality of $V_j$, $j=1,2$.

\begin{lemma}\label{lem:fbBbound}
  The number of $(k-1)$-faces of $\bB$ is bounded from above as
  follows:
  \begin{equation}\label{equ:fbBbound}
    f_{k-1}(\bB)\le
    \sum_{j=1}^{k-1}
    \binom{n_1}{j}\binom{n_2}{k-j}
    =\binom{n_1+n_2}{k}-\binom{n_1}{k}-\binom{n_2}{k},
    \quad 1\le{}k\le{}d,
  \end{equation}
  where equality holds if and only if $\cC$ is
  \bn{k}.
\end{lemma}

\newcommand{\subf}[1]{{#1}'}

\begin{proof}
  The case $k=1$ is trivial. We have
  $f_0(\bB)=0=\tbinom{n_1+n_2}{1}-\tbinom{n_1}{1}-\tbinom{n_2}{1}$,
  since $\bB$ does not contain any $0$-faces of $\cC$: the $0$-faces
  of $\cC$, \ie the vertices of $\cC$, are either vertices of $\cC_1$
  or $\cC_2$.

  Let $k\ge{}2$, and denote by $V_F$ the subset of $V$ defining a face
  $F\in\bB$. Define $\varphi_{k-1}:\bB\to{}2^V$ to be the mapping
  that maps a $(k-1)$-face $F\in\bB$ to a subset $V_F'$ of $V_F$, of
  size $k$, such that:
  \begin{enumerate}
  \item[(1)] $V_F'$ is $(k-1)$-dimensional, and
  \item[(2)] $V_F'\cap{}V_j\ne\emptyset$, $j=1,2$.
  \end{enumerate}
  The mapping $\varphi_{k-1}$ is well defined in the sense that
  such a subset $V_F'$ always exists. We are going to show this by
  induction on $k$. For $k=2$, simply choose $V_F'=\{v_1,v_2\}$, where
  $v_1\in{}V_F\cap{}V_1$ and $v_2\in{}V_F\cap{}V_2$.
  Suppose that our claim holds for $k\ge{}2$, \ie for
  any $(k-1)$-face $F$ of $\bB$, there exists a subset $V_F'$ of
  $V_F$ of size $k$, such that $V_F'$ is $(k-1)$-dimensional, and
  $V_F'\cap{}V_j\ne\emptyset$, $j=1,2$. We wish to show that this is
  also true for $k+1$. Indeed, let $F$ be a $k$-face of $\bB$. If $F$ is
  defined by $k+1$ vertices (\ie $F$ is simplicial), $V_F'$ is simply
  $V_F$. Clearly, $V_F$ is $k$-dimensional, and
  $V_F\cap{}V_j\ne\emptyset$, $j=1,2$, since $F$ is a $k$-face of $\bB$.
  Otherwise, suppose $F$ is defined by more than $k+1$ vertices,
  \ie $|V_F|>k+1$. Consider the $(k-1)$-faces of $F$: at least
  one of these faces has to be a face in $\bB$ (since, otherwise, $F$
  would not have been a face of $\bB$, but rather a face of either
  $\cC_1$ or $\cC_2$), and let $\subf{F}$ be such a $(k-1)$-face of $F$.
  By the induction hypothesis there exists a subset $V_{\subf{F}}'$ of
  $V_{\subf{F}}$ of size $k$, such that $V_{\subf{F}}'$ is
  $(k-1)$-dimensional and $V_{\subf{F}}'\cap{}V_j\ne\emptyset$,
  $j=1,2$. But then there exists a vertex $v\in{}V_F\sm{}V_{\subf{F}}'$,
  such that the set $V_F'=V_{\subf{F}}'\cup\{v\}$ is $k$-dimensional
  (if this is not the case, then $F$ would have been $(k-1)$-dimensional,
  which contradicts the fact that $F$ is a $k$-face of $\bB$). The set
  $V_F'$ is the set we were looking for: $V_F'$ has size $k+1$ (since
  $|V_{\subf{F}}'|=k$), $V_F'$ is $k$-dimensional (we just argued
  that), and $V_F'\cap{}V_j\ne\emptyset$, $j=1,2$ (this holds for
  $V_{\subf{F}}'$, and, thus, it holds for $V_F'$ as well).

  We argue that the mapping $\varphi_{k-1}$ is an injection
  from the faces of $\bB$ to the subsets of size $k$ of $V$ which
  contain elements from both $V_1$ and $V_2$. To this end, consider two
  $(k-1)$-faces $F_1$ and $F_2$ of $\bB$, such that $F_1\ne{}F_2$, and
  assume that $\varphi_{k-1}(F_1)=\varphi_{k-1}(F_2)$.
  Since $\varphi_{k-1}(F_1)=\varphi_{k-1}(F_2)$, we have that
  $V_{F_1}'=V_{F_2}'$ and both $V_{F_1}'$ and $V_{F_2}'$ are
  $(k-1)$-dimensional. Therefore, the intersection $F_1\cap{}F_2$ is
  not only a face $F$ of both $F_1$ and $F_2$, but also contains all
  vertices in $V_{F_1}'=V_{F_2}'$. Since $V_{F_1}'$, or $V_{F_2}'$, is
  $(k-1)$-dimensional, $F$ is a $(k-1)$-face of both $F_1$ and
  $F_2$. On the other hand, the only $(k-1)$-face of either $F_1$, or
  $F_2$, is $F_1$, or $F_2$, respectively. Hence $F=F_1$ and
  $F=F_2$, that is $F_1=F_2$, which contradicts our assumption that
  $F_1\ne{}F_2$. Summarizing, we have that if $F_1\ne{}F_2$, then
  $\varphi_{k-1}(F_1)\ne\varphi_{k-1}(F_2)$, \ie the mapping
  $\varphi_{k-1}$ is an injection.

  Having established that $\varphi_{k-1}:\bB\to{}2^V$ is an injection,
  we proceed with the upper bound and equality claim of the lemma.
  The number of the subsets of $V$ of size $k$, that have at least one
  vertex from both $V_1$ and $V_2$ is precisely
  $\sum_{1\le{}j\le{}k-1}\binom{n_1}{j}\binom{n_2}{k-j}$, which is
  equal to $\binom{n_1+n_2}{k}-\binom{n_1}{k}-\binom{n_2}{k}$,
  according to Vandermonde's convolution identity. This gives the
  upper bound. Furthermore, notice that the injection $\varphi_{k-1}$
  becomes a bijection if and only if for every non-empty subset $S_1$
  of $V_1$ and every non-empty subset $S_2$ of $V_2$, where
  $|S_1|+|S_2|=k$, the vertex set $S_1\cup{}S_2$ defines a
  $(k-1)$-face of $\cC$. In other words, equality in
  \eqref{equ:fbBbound} can only hold if and only if $\cC$ is \bn{k}.
\end{proof}

Combining Lemma \ref{lem:bn2n} with Lemma \ref{lem:fbBbound} we deduce
that, if the inequality in Lemma \ref{lem:fbBbound} holds as
equality for some $l$, then we also have
$f_{k-1}(\bB)=\binom{n_1+n_2}{k}-\binom{n_1}{k}-\binom{n_2}{k}$ for
all $k$ with $1\le{}k\le{}l-1$.

\section{Upper bounds}
\label{sec:ub}

\renewcommand{\labelsize}{\footnotesize}
\begin{figure}[!b]
  \begin{center}
    \psfrag{P1}[][]{\labelsize$P_1$}
    \psfrag{P2}[][]{\labelsize$P_2$}
    \psfrag{Pi1}[][]{\labelsize$\Pi_1$}
    \psfrag{Pi2}[][]{\labelsize$\Pi_2$}
    \psfrag{P~}[][]{\labelsize$\tilde{P}$}
    \psfrag{Pi~}[][]{\labelsize$\tilde\Pi$}
    \psfrag{F}[][]{\labelsize$\fF$}
    \includegraphics[width=0.9\textwidth]{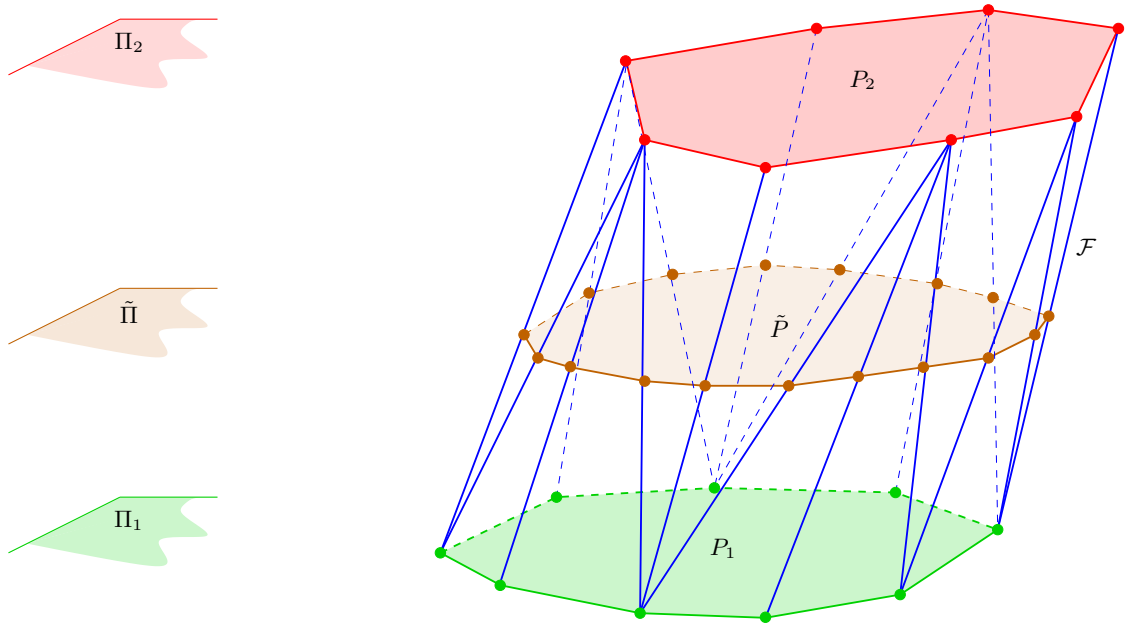}
  \end{center}
  \caption{The $d$-polytopes $P_1$ and $P_2$ are embedded in the
    hyperplanes $\Pi_1=\{x_{d+1}=0\}$ and $\Pi_2=\{x_{d+1}=0\}$ of
    $\euc^{d+1}$. The polytope $\tilde{P}$ is the intersection of
    $CH_{d+1}(\{P_1,P_2\})$ with the hyperplane
    $\tilde\Pi=\{x_{d+1}=\lambda\}$.}\label{fig:fig1}
\end{figure}

Let $P_1$ and $P_2$ be two $d$-polytopes in $\euc^d$, with $n_1$ and
$n_2$ vertices, respectively. The Minkowski sum $P_1\oplus{}P_2$ of
$P_1$ and $P_2$ is the $d$-polytope
$P_1\oplus{}P_2=\{p+q\,|\,p\in{}P_1,q\in{}P_2\}$,
whereas their weighted Minkowski sum is defined as
$(1-\lambda)P_1\oplus\lambda{}P_2=
\{(1-\lambda)p+\lambda{}q\,|\,p\in{}P_1,q\in{}P_2\}$, where
$\lambda\in(0,1)$.
Let us embed $P_1$ (\resp $P_2$) in the hyperplane $\Pi_1$ (\resp
$\Pi_2$) of $\euc^{d+1}$ with equation $\{x_{d+1}=0\}$ (\resp
$\{x_{d+1}=1\}$). Then the weighted Minkowski sum
$(1-\lambda)P_1\oplus\lambda{}P_2$ is the $d$-polytope we
get when intersecting $CH_{d+1}(\{P_1,P_2\})$ with the hyperplane
$\{x_{d+1}=\lambda\}$ (see Fig. \ref{fig:fig1}).
From this reduction it is evident that the weighted Minkowski sum
$(1-\lambda)P_1\oplus\lambda{}P_2$, $\lambda\in(0,1)$, does not really
depend on the specific value of $\lambda$, in the sense that the
weighted Minkowski sums of $P_1$ and $P_2$ for two different $\lambda$
values are combinatorially equivalent. Furthermore, the weighted
Minkowski sum of $P_1$ and $P_2$ is also combinatorially equivalent to
the unweighted Minkowski sum $P_1\oplus{}P_2$, since $P_1\oplus{}P_2$
is nothing but $\frac{1}{2}P_1\oplus\frac{1}{2}P_2$, scaled by a
factor of 2. In view of these observations, in the rest of the paper
we focus on the sum $P_1\oplus{}P_2$, with the understanding that
our results carry over to the weighted Minkowski sum
$(1-\lambda)P_1\oplus\lambda{}P_2$, for any $\lambda\in(0,1)$.

As in the previous paragraph, let $\Pi_1$ and $\Pi_2$ be the
hyperplanes $\{x_{d+1}=0\}$ and $\{x_{d+1}=1\}$, and let $\tilde{\Pi}$
be a hyperplane in $\euc^{d+1}$ parallel and in-between $\Pi_1$ and
$\Pi_2$. Consider two $d$-polytopes $P_1$ and $P_2$ embedded in
$\euc^{d+1}$, and in the hyperplanes $\Pi_1$ and $\Pi_2$,
respectively, and call $P$ the convex hull
$CH_{d+1}(\{P_1,P_2\})$. Karavelas and Tzanaki
\cite[Lemma 2]{kt-chsch-11b} have shown that the vertices of $P_1$ and
$P_2$ can be perturbed in such a way that:
\begin{enumerate}
\item the vertices of $P_1'$ and $P_2'$ remain in $\Pi_1$ and
  $\Pi_2$, respectively, and both $P_1'$ and $P_2'$ are simplicial,
\item $P'=CH_{d+1}(\{P_1',P_2'\})$ is also simplicial, except possibly
  the facets $P_1'$ and $P_2'$, and
\item the number of vertices of $P_1'$ and $P_2'$ is the same as the
  number of vertices of $P_1$ and $P_2$, respectively, whereas
  $f_k(P)\le{}f_k(P')$ for all $k\ge{}1$,
\end{enumerate}
where $P_1'$ and $P_2'$ are the polytopes in $\Pi_1$ and $\Pi_2$ we
get after perturbing the vertices of $P_1$ and $P_2$, respectively.
In view of this result, it suffices to consider the case where both
$P_1$, $P_2$ and their convex hull $P=CH_{d+1}(\{P_1,P_2\})$ are
simplicial complexes (except possibly the facets $P_1$ and $P_2$ of
$P$). In the rest of this section, we consider that this is the
case: $P$ is considered simplicial, with the possible exception of
its two facets $P_1$ and $P_2$.
Let $\fF$ be the set of proper faces of $P$ having non-empty
intersection with $\tilde{\Pi}$. Note that
$\tilde{P}=P\cap\tilde{\Pi}$ is a $d$-polytope, which is, in
general, non-simplicial, and whose proper non-trivial faces
are intersections of the form $F\cap\tilde\Pi$ where
$F\in\fF$. As we have already observed above,
$\tilde{P}$ is combinatorially equivalent to the Minkowski sum
$P_1\oplus{}P_2$. Furthermore,
\begin{equation}\label{equ:facenumbercorresp}
  f_{k-1}(P_1\oplus{}P_2)=f_{k-1}(\tilde{P})=f_k(\fF),
  \qquad 1\le{}k\le{}d.
\end{equation}
The rest of this section is devoted to deriving upper bounds for
$f_k(\fF)$, which, by relation \eqref{equ:facenumbercorresp}, become
upper bounds for $f_{k-1}(P_1\oplus{}P_2)$.

Let $\kK$ be the polytopal complex whose faces are all the faces of $\fF$,
as well as the faces of $P$ that are subfaces of faces in $\fF$.
It is easy to see that the $d$-faces of $\kK$ are exactly the
$d$-faces of $\fF$, and, thus, $\kK$ is a pure simplicial
$d$-complex, with the $d$-faces of $\fF$ being the facets of $\kK$.
Moreover, the set of $k$-faces of $\kK$ is the disjoint union of
the sets of $k$-faces of $\fF$, $\bx{}P_1$ and $\bx{}P_2$. This implies:
\begin{equation}\label{equ:fkK}
  f_k(\kK)=f_k(\fF)+f_k(\bx{}P_1)+f_k(\bx{}P_2),\qquad -1\le{}k\le{}d.
\end{equation}
where $f_d(\bx{}P_j)=0$, $j=1,2$, and conventionally we set
$f_{-1}(\fF)=-1$.

\renewcommand{\labelsize}{\footnotesize}
\begin{figure}[t]
  \begin{center}
    \psfrag{P1}[][]{\labelsize$\bx{}P_1$}
    \psfrag{P2}[][]{\labelsize$\bx{}P_2$}
    \psfrag{Pi1}[][]{\labelsize$\Pi_1$}
    \psfrag{Pi2}[][]{\labelsize$\Pi_2$}
    \psfrag{P~}[][]{\labelsize$\tilde{P}$}
    \psfrag{Pi~}[][]{\labelsize$\tilde\Pi$}
    \psfrag{F}[][]{\labelsize$\fF$}
    \psfrag{y1}[][]{\labelsize$y_1$}
    \psfrag{y2}[][]{\labelsize$y_2$}
    \includegraphics[width=0.9\textwidth]{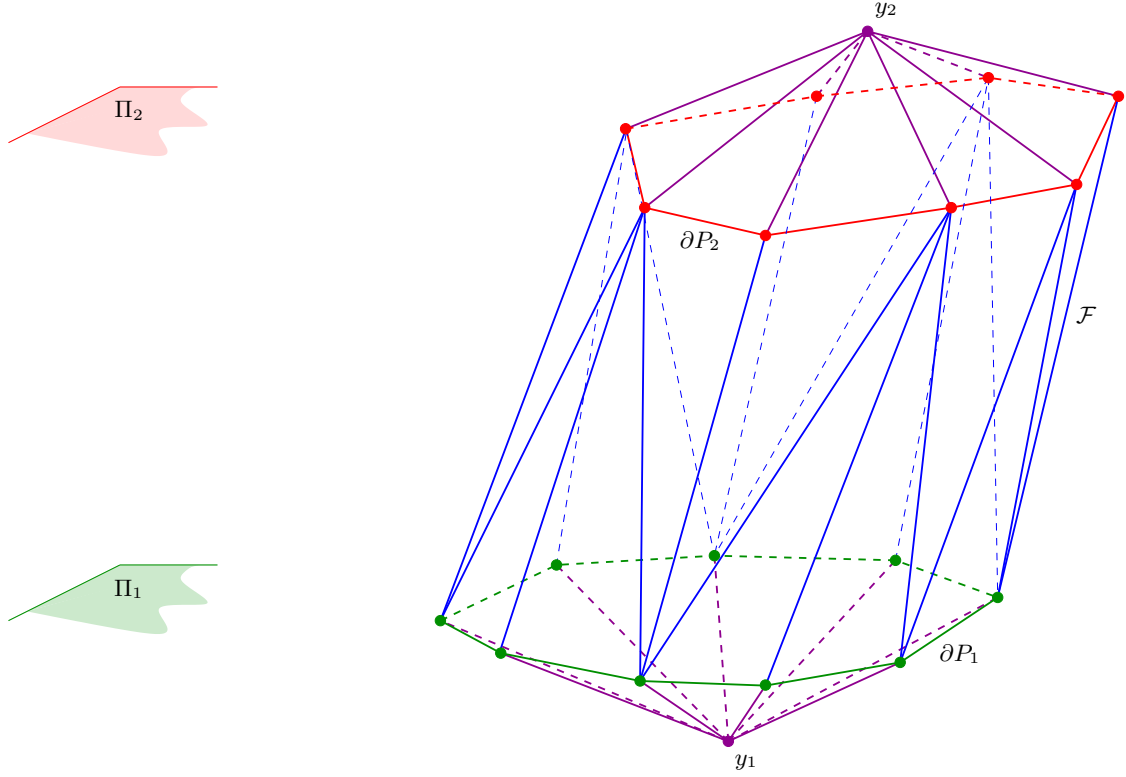}
  \end{center}
  \caption{The polytope $Q$ is created by adding two vertices $y_1$
    and $y_2$. The vertex $y_1$ (\resp $y_2$) is below $P_1$ (\resp
    above $P_2$), and is visible by the vertices of $P_1$ (\resp $P_2$)
    only.}\label{fig:fig2}
\end{figure}

Let $y_1$ (\resp $y_2$) be a point below $\Pi_1$ (\resp above $\Pi_2$),
such that the vertices of $P_1$ (\resp $P_2$) are the only
vertices of $P$ visible from $y_1$ (\resp $y_2$) (see Fig. \ref{fig:fig2}).
To achieve this, we choose $y_1$ (\resp $y_2$) to be a point beyond
the facet $P_1$ (\resp $P_2$) of $P$, and beneath every other facet of
$P$.
Let $Q$ be the $(d+1)$-polytope that is the convex hull of the
vertices of $P_1$, $P_2$, $y_1$ and $y_2$. 
Observe that the faces of $\bx{}P$ (and thus all faces of $\fF$), except
for the facets $P_1$ and $P_2$ of $\bx{}P$, are all faces of the
boundary complex $\bx{}Q$. To see
that, notice that a supporting hyperplane $H_F$ for a facet $F\in{}P$,
with $F\ne{}P_1,P_2$, is also a supporting hyperplane for $Q$. Indeed,
the vertices of $F$ are vertices of $Q$ different from $y_1$ and $y_2$
and thus, every vertex of $P$ that is not a vertex of $F$ strictly
satisfies all hyperplane inequalities for $P$.
Also, by construction, the points $y_1$ and $y_2$ strictly satisfy all
hyperplane inequalities apart from those for $\Pi_1$ and $\Pi_2$,
respectively. Since $H_F$ is a hyperplane other than $\Pi_1$ and
$\Pi_2$ we deduce that all vertices of $P$, as well as $y_1$ and
$y_2$, lie on the same halfspace defined by $H_F$, and therefore $H_F$
supports $Q$.
The faces of $Q$ that are not faces of $\fF$ are the faces in the star
$\sS_1$ of $y_1$ and the star $\sS_2$ of $y_2$. To verify this,
consider a $k$-face $F$ of $P_1$, and let $F'$ be a face in $\fF$ that
contains $F$. Let $H'$ be a supporting hyperplane of $F'$ with respect
to $P$. Tilt $H'$ until it hits the point $y_1$, while keeping $H'$
incident to $F'$, and call $H''$ this tilted hyperplane. $H''$ is
a supporting hyperplane for $y_1$ and the vertex set of $P_1$, and
thus is a supporting hyperplane for $Q$. The same argument can be
applied for $\str[y_2]{Q}$. In fact, the boundary complex
$\bx{}P_1$ of $P_1$ (\resp $\bx{}P_2$ of $P_2$) is nothing
but the link of $y_1$ (\resp $y_2$) in $Q$. 

It is easy to realize that the set of $k$-faces of $\bx{}Q$ is
the disjoint union of the $k$-faces of $\fF$, $\sS_1$ and
$\sS_2$. This implies that:
\begin{equation}\label{equ:fkQ}
  f_k(\bx{}Q)=f_k(\fF)+f_k(\sS_1)+f_k(\sS_2),\quad 0\le{}k\le{}d,
\end{equation}
where $f_0(\fF)=0$.
The $k$-faces of $\bx{}Q$ in $\sS_j$ are either $k$-faces of
$\bx{P_j}$ or $k$-faces defined by $y_j$ and a $(k-1)$-face of
$\bx{P_j}$. In fact, there exists a bijection between the
$(k-1)$-faces of $\bx{}P_j$ and the $k$-faces of $\sS_j$
containing $y_j$. Hence, we have, for $j=1,2$:
\begin{equation}\label{equ:fkS}
  f_k(\sS_j)=f_k(\bx{}P_j)+f_{k-1}(\bx{}P_j),\quad 0\le{}k\le{}d,
\end{equation}
where $f_{-1}(\bx{}P_j)=1$ and $f_d(\bx{}P_j)=0$.
Combining relations \eqref{equ:fkQ} and \eqref{equ:fkS}, we get:
\begin{equation}\label{equ:fkQ1}
  f_k(\bx{}Q)=f_k(\fF)+f_k(\bx{}P_1)+f_{k-1}(\bx{}P_1)+
  f_k(\bx{}P_2)+f_{k-1}(\bx{}P_2),\quad 0\le{}k\le{}d.
\end{equation}

We call $\kK_j$, $j=1,2$, the subcomplex of $\bx{}Q$ consisting
of either faces of $\kK$ or faces of $\sS_j$. $\kK_j$ is a pure
simplicial $d$-complex the facets of which are either facets in the
star $\sS_j$ of $y_j$ or facets of $\kK$. Furthermore, $\kK_j$ is
shellable. To see this first notice that $\bx{}Q$
is shellable ($Q$ is a polytope). Consider a line shelling
$F_1,F_2,\ldots,F_s$ of $\bx{}Q$ that shells $\str[y_2]{\bx{}Q}$
last, and let $F_{\lambda+1},F_{\lambda+2},\ldots,F_s$ be the facets of
$\bx{Q}$ that correspond to $\sS_2$. Trivially, the subcomplex of
$\bx{Q}$, the facets of which are $F_1,F_2,\ldots,F_\lambda$, is
shellable; however, this subcomplex is nothing but $\kK_1$. The
argument for $\kK_2$ is analogous.

Notice that $Q$ is a simplicial $(d+1)$-polytope, while $\kK$, $\kK_1$
and $\kK_2$ are simplicial $d$-complexes; hence their $h$-vectors
are well defined. More precisely:
\begin{equation}\label{equ:hQ}
  h_k(\yY)
  =\sum_{i=0}^{k}(-1)^{k-i}\binom{d+1-i}{d+1-k}f_{i-1}(\yY),
  \quad 0\le{}k\le{}d+1,
\end{equation}
where $\yY$ stands for either $\bx{}Q$, $\kK$, $\kK_1$ or $\kK_2$.
We define the $f$-vector of $\fF$ to be the $(d+2)$-vector
$\mb{f}(\fF)=(f_{-1}(\fF),f_0(\fF),\ldots,f_d(\fF))$, where recall that
$f_{-1}(\fF)=-1$, and from this we can also define the $(d+2)$-vector
$\mb{h}(\fF)=(h_0(\fF),h_1(\fF),\ldots,$ $h_{d+1}(\fF))$, where
\begin{equation}\label{equ:hF}
  h_k(\fF)
  =\sum_{i=0}^{k}(-1)^{k-i}\binom{d+1-i}{d+1-k}f_{i-1}(\fF),
  \quad 0\le{}k\le{}d+1.
\end{equation}
We call this vector the $h$-vector of $\fF$. As for polytopal
complexes and polytopes, the $f$-vector of $\fF$ defines the
$h$-vector of $\fF$ and vice versa. In particular, solving the
defining equations \eqref{equ:hF} of the elements of $\mb{h}(\fF)$ in
terms of the elements of $\mb{f}(\fF)$ we get:
\begin{equation}\label{equ:facesF}
  f_{k-1}(\fF)=\sum_{i=0}^{d+1}\binom{d+1-i}{k-i}h_i(\fF),
  \qquad 0\le{}k\le{}d+1.
\end{equation}

The next lemma
associates the elements of $\mb{h}(\bx{}Q)$, $\mb{h}(\kK)$,
$\mb{h}(\kK_1)$, $\mb{h}(\kK_2)$, $\mb{h}(\fF)$,
$\mb{h}(\bx{}P_1)$ and $\mb{h}(\bx{}P_2)$. The last among the
relations in the lemma can be thought of as the analogue of the
Dehn-Sommerville equations for $\fF$.

\begin{lemma}
  For all $0\le{}k\le{}d+1$ we have:
  \begin{align}
    h_k(\bx{}Q)&=h_k(\fF)+h_k(\bx{P_1})+h_k(\bx{P_2}),\label{equ:hkF}\\
    h_k(\kK)&=h_k(\fF)+g_k(\bx{P_1})+g_k(\bx{P_2}),\label{equ:hkK}\\
    h_k(\kK_j)&=h_k(\kK)+h_{k-1}(\bx{P_j}),\qquad\qquad\qquad j=1,2,
    \label{equ:hkKj}\\
    h_{d+1-k}(\fF)&=h_k(\fF)+g_k(\bx{}P_1)+g_k(\bx{}P_2).\label{equ:hFDS}
  \end{align}
\end{lemma}

\begin{proof}
  Let $\yY$ denote either $\fF$ or a pure simplicial subcomplex of
  $\bx{}Q$. We define the operator $\SSS_{k}(\cdot;\delta,\nu)$ whose
  action on $\yY$ is as follows:
  \begin{equation}
    \SSS_{k}(\yY;\delta,\nu)
    =\sum_{i=1}^{\delta}(-1)^{k-i}\binom{\delta-i}{\delta-k}f_{i-\nu}(\yY).
  \end{equation}
  It is easy to verify\footnote{See Section \ref{app:sumop} of the Appendix
    for detailed derivations.} that if $\yY$ is $\delta$-dimensional (this
  includes the case $\yY\equiv\fF$), then
  \begin{equation}\label{equ:sum-d-1-a}
    \SSS_{k}(\yY;\delta,1)=h_k(\yY)-(-1)^k\binom{\delta}{\delta-k}f_{-1}(\yY).
  \end{equation}
  while if $\yY$ is $(\delta-1)$-dimensional, then
  \begin{align}
    \SSS_{k}(\yY;\delta,1)
    &=h_k(\yY)-h_{k-1}(\yY)-(-1)^k\binom{\delta}{\delta-k}f_{-1}(\yY),
    \text{ and}\label{equ:sum-d-1-b}\\
    \SSS_{k}(\yY;\delta,2)&=h_{k-1}(\yY).\label{equ:sum-d-2}
  \end{align}
  Applying the operator $\SSS_k(\cdot;d+1,1)$ to $\bx{}Q$ and using
  relation \eqref{equ:fkQ1} we get:
  \begin{equation}\label{equ:sums}
    \begin{aligned}
      \SSS_{k}(\bx{}Q;d+1,1)=\ &\SSS_k(\fF;d+1,1)
      +\SSS_k(\bx{}P_1;d+1,1)+\SSS_k(\bx{}P_1;d+1,2)\\
      &+\SSS_k(\bx{}P_2;d+1,1)+\SSS_k(\bx{}P_2;d+1,2).
    \end{aligned}
  \end{equation}
  Substituting in \eqref{equ:sums}, using relations
  \eqref{equ:sum-d-1-a}-\eqref{equ:sum-d-2}, we get:
  \begin{equation*}
    \begin{aligned}
      h_k(\bx{}Q)-(-1)^k&\binom{d+1}{d+1-k}f_{-1}(\bx{}Q)
      =\left[h_k(\fF)-(-1)^k\binom{d+1}{d+1-k}f_{-1}(\fF)\right]\\
      &+\left[h_k(\bx{}P_1)-h_{k-1}(\bx{}P_1)
        -(-1)^k\binom{d+1}{d+1-k}f_{-1}(\bx{}P_1)\right]
      +h_{k-1}(\bx{}P_1)\\
      &+\left[h_k(\bx{}P_2)-h_{k-1}(\bx{}P_2)
        -(-1)^k\binom{d+1}{d+1-k}f_{-1}(\bx{}P_2)\right]
      +h_{k-1}(\bx{}P_2).
    \end{aligned}
  \end{equation*}
  Given that $f_{-1}(\bx{}Q)=f_{-1}(\bx{}P_1)=f_{-1}(\bx{}P_2)=1$, and
  $f_{-1}(\fF)=-1$, the above equality simplifies to relation
  \eqref{equ:hkF}.

  Recall that the set of $k$-faces of $\kK$ is the disjoint union
  of the $k$-faces of $\fF$, the $k$-faces of $\bx{}P_1$, and the
  $k$-faces of $\bx{}P_2$. Applying the operator
  $\SSS_{k}(\cdot;d+1,1)$ to $\kK$, and using relation \eqref{equ:fkK}
  we get:
  \begin{equation*}
    h_k(\kK)=h_k(\fF)+h_k(\bx{P_1})-h_{k-1}(\bx{P_1})
    +h_k(\bx{P_2})-h_{k-1}(\bx{P_2}),\qquad 0\le{}k\le{}d+1,
  \end{equation*}
  which reduces to relation \eqref{equ:hkK} if we replace the
  difference $h_k(\cdot)-h_{k-1}(\cdot)$ by the corresponding element
  of $\mb{g}(P_j)$.

  The $k$-faces of $\kK_j$, $j=1,2$, are either $k$-faces of
  $\kK$ or $k$-faces of the star $\sS_j$ of $y_j$ that contain
  $y_j$. The latter faces are in one-to-one correspondence with the
  $(k-1)$-faces of $\bx{}P_j$, i.e., we get:
  \begin{equation}\label{equ:fkKj}
    f_k(\kK_j)=f_k(\kK)+f_{k-1}(\bx{}P_j),\qquad 0\le{}k\le{}d.
  \end{equation}
  Once again, applying the operator $\SSS_{k}(\cdot;d+1,1)$ to
  $\kK_j$, and using relation \eqref{equ:fkKj} we get relation
  \eqref{equ:hkKj}.

  We end the proof of this lemma by proving relations \eqref{equ:hFDS}.
  Since $Q$ is a simplicial $(d+1)$-polytope, and $P_1$, $P_2$ are
  simplicial $d$-polytopes, the Dehn-Sommerville equations for these
  polytopes hold. More precisely:
  \begin{equation}\label{equ:DSF1}
    \begin{aligned}
      h_{d+1-k}(\bx{}Q)&=h_k(\bx{}Q), &\quad{}0\le{}k\le{}d+1,\\
      h_{d-k}(\bx{}P_j)&=h_k(\bx{}P_j),
      & 0\le{}k\le{}d,&\qquad j=1,2.
    \end{aligned}
  \end{equation}
  Combining the above relations with \eqref{equ:hkF} we get, for all
  $0\le{}k\le{}d+1$:
  \begin{equation}
    h_{d+1-k}(\fF)+h_{d+1-k}(\bx{}P_1)+h_{d+1-k}(\bx{}P_2)
    =h_k(\fF)+h_k(\bx{}P_1)+h_k(\bx{}P_2),
  \end{equation}
  or, equivalently:
  \begin{equation}
    h_{d+1-k}(\fF)+h_{k-1}(\bx{}P_1)+h_{k-1}(\bx{}P_2)
    =h_k(\fF)+h_k(\bx{}P_1)+h_k(\bx{}P_2),
  \end{equation}
  which finally gives:
  \begin{equation*}
    h_{d+1-k}(\fF)=h_k(\fF)+g_k(\bx{}P_1)+g_k(\bx{}P_2).
  \end{equation*}
  In the equations above, 
  $g_{0}(\bx{}P_j)=-g_{d+1}(\bx{}P_j)=1$, $j=1,2$.
\end{proof}

Recall that the main goal in this section is to derive upper bounds
for the elements of $\mb{h}(\fF)$. The most critical step toward
this goal is the recurrence inequality for the elements of
$\mb{h}(\fF)$ described in the following lemma.

\begin{lemma}\label{lem:hFrecur}
  For all $0\le{}k\le{}d$,
  \begin{equation}\label{equ:hkFrecur}
    h_{k+1}(\fF)\le\frac{n_1+n_2-d-1+k}{k+1}\,h_k(\fF)
    +\frac{n_1}{k+1}\,g_k(\bx{P_2})+\frac{n_2}{k+1}\,g_k(\bx{P_1}).
  \end{equation}
\end{lemma}

\begin{proof}
  Let us denote by $V$ the vertex set of $\bx{Q}$, and by $V_j$ the
  vertex set of $\bx{P_j}$, $j=1,2$. Let $\link{\yY}$ be a shorthand
  for $\linkfull{\yY}$, where $v$ is a vertex of $\yY$, and $\yY$
  stands for either $\kK_1$ or $\kK_2$, or the boundary complex
  of a simplicial polytope.

  McMullen \cite{m-mnfcp-70} in his original proof of the Upper Bound
  Theorem for polytopes proved that for any $d$-polytope $P$ the
  following relation holds:
  \begin{equation}\label{equ:hkrelP}
    (k+1)h_{k+1}(\bx{P})+(d-k)h_k(\bx{P})
    =\sum_{v\in\text{vert}(\bx{P})}h_k(\link{\bx{P}}),\qquad 0\le{}k\le{}d-1.
  \end{equation}
  Furthermore, we have $h_k(\link{\bx{}P})\le{}h_k(\bx{}P)$. To see this
  consider a shelling of $\bx{}P$ that shells $\str{\bx{}P}$ first. The
  contributions to $h_k(\bx{}P)$ coincide with the contributions to
  $h_k(\link{\bx{}P})$ during the shelling of $\str{\bx{}P}$. After
  the shelling has left $\str{\bx{}P}$ we get no more contributions to
  $h_k(\link{\bx{}P})$, whereas we may get contributions to
  $h_k(\bx{}P)$. Therefore:
  \begin{equation}\label{equ:sumlinkbound}
    \sum_{v\in\text{vert}(\bx{}P)}h_k(\link{\bx{}P})\le{}f_0(\bx{}P)h_k(\bx{}P),
    \qquad 0\le{}k\le{}d-1.
  \end{equation}

  Applying relation \eqref{equ:hkrelP} to $Q$, $P_1$ and
  $P_2$ we get the following relations:
  \begin{align}
    \label{equ:hkQrel}
    (k+1)h_{k+1}(\bx{Q})+(d+1-k)h_k(\bx{Q})
    &=\sum_{v\in{}V}h_k(\link{\bx{Q}}), &0\le{}k\le{}d.\\
    \label{equ:hkP1rel}
    (k+1)h_{k+1}(\bx{P_1})+(d-k)h_k(\bx{P_1})
    &=\sum_{v\in{}V_1}h_k(\link{\bx{P_1}}), &0\le{}k\le{}d-1.\\
    \label{equ:hkP2rel}
    (k+1)h_{k+1}(\bx{P_2})+(d-k)h_k(\bx{P_2})
    &=\sum_{v\in{}V_2}h_k(\link{\bx{P_2}}), &0\le{}k\le{}d-1.
  \end{align}
  Recall that the link of $y_j$ in $\bx{Q}$ is $\bx{P_j}$, $j=1,2$,
  and observe that the link  of $v\in{}V_j$ in $\bx{}Q$ coincides with
  $\link{\kK_j}$.
  Expanding relation \eqref{equ:hkQrel} by means of relation
  \eqref{equ:hkF}, 
  we deduce:
  \begin{equation}
    \begin{aligned}
      (k+1)&[h_{k+1}(\fF)+h_{k+1}(\bx{P_1})+h_{k+1}(\bx{P_2})]
      +(d+1-k)[h_k(\fF)+h_k(\bx{P_1})+h_k(\bx{P_2})]=\\
      &=(k+1)h_{k+1}(\fF)+(d+1-k)h_k(\fF)
      +(k+1)h_{k+1}(\bx{P_1})+(d-k)h_k(\bx{P_1})\\
      &\qquad
      +(k+1)h_{k+1}(\bx{P_2})+(d-k)h_k(\bx{P_2})
      +h_k(\bx{P_1})+h_k(\bx{P_2})\\
      &=\sum_{v\in{}V}h_k(\link{\bx{Q}})
      =h_k(\link[y_1]{\bx{Q}})+h_k(\link[y_2]{\bx{Q}})
      +\sum_{v\in{}V_1\cup{}V_2}h_k(\link{\bx{Q}})\\
      &=h_k(\bx{P_1})+h_k(\bx{P_2})
      +\sum_{v\in{}V_1}h_k(\link{\kK_1})+\sum_{v\in{}V_2}h_k(\link{\kK_2}).
    \end{aligned}
  \end{equation}
  Utilizing relations \eqref{equ:hkP1rel} and \eqref{equ:hkP2rel}, the
  above equation is equivalent to:
  \begin{equation}\label{equ:hkFrecur1}
    (k+1)h_{k+1}(\fF)+(d+1-k)h_k(\fF)
    =\sum_{v\in{}V_1}[h_k(\link{\kK_1})-h_k(\link{\bx{}P_1})]
    +\sum_{v\in{}V_2}[h_k(\link{\kK_2})-h_k(\link{\bx{}P_2})].
  \end{equation}

  Let us now consider a vertex $v\in{}V_1$, and a shelling
  $\shell{\bx{}Q}$ of $\bx{}Q$ that shells $\str{\bx{}Q}$ first and
  $\str[y_2]{\bx{}Q}$ last. Such a shelling does exit: consider a
  point $v'$ (\resp $y_2'$) beyond $v$ (\resp $y_2$) such that the
  line $\ell$ defined by $v'$ and $y_2'$ does not pass through $v$ and
  $y_2$. Call $v''$ and $y_2''$ the points of intersection of $\ell$
  with $\bx{}Q$, and notice that, since $v$ and $y_2$ are not visible
  to each other, the only points of intersection of $\ell$ with
  $\bx{}Q$ are the points $v''$ and $y_2''$. 
  The shelling $\shell{\bx{}Q}$ is the line shelling of $\bx{}Q$
  induced by $\ell$ when we move from $v''$ away from $\bx{}Q$ towards
  $+\infty$, and then from $-\infty$ to $y_2''$.
  Notice that $\shell{\bx{}Q}$ induces a shelling $\shell{\kK_1}$ for
  $\kK_1$ that shells $\str{\kK_1}$ first (any shelling of $\bx{}Q$,
  that shells $\str[y_2]{\bx{}Q}$ last, induces a shelling for
  $\kK_1$, where the order of the facets of $\kK_1$ in this shelling
  is the same as their order in the shelling of $\bx{}Q$).
  On the other hand, $\shell{\kK_1}$ also induces
  (cf. \cite[Lemma 8.7]{z-lp-95}):
  \begin{enumerate}
  \item
    a shelling $\shell{\link{\kK_1}}$ for $\link{\kK_1}$, and
  \item
    a shelling $\shell{\bx{}P_1}$ for $\bx{}P_1$ that shells
    $\str{\bx{}P_1}$ first (recall that
    $\bx{}P_1\equiv\link[y_1]{\bx{}Q}\equiv\link[y_1]{\kK_1}$),
  \end{enumerate}
  while $\shell{\bx{}P_1}$ induces a shelling
  $\shell{\link{\bx{}P_1}}$ for $\link{\bx{}P_1}$ (again,
  cf. \cite[Lemma 8.7]{z-lp-95}).
  The interested reader may refer to
  Figs. \ref{fig:K1shelling1}--\ref{fig:K1+P1-shellings2},
  where we show a shelling $\shell{\kK_1}$ of $\kK_1$ that shells
  $\str{\kK_1}$ first, along with the induced shellings
  $\shell{\link{\kK_1}}$ and $\shell{\bx{}P_1}$.
  In particular, Figs. \ref{fig:K1shelling1}--\ref{fig:K1shelling3}
  show the step-by-step construction of $\kK_1$ from
  $\shell{\kK_1}$.
  Fig. \ref{fig:K1+linkv-shellings} shows the
  step-by-step construction of $\str{\kK_1}$ from $\shell{\kK_1}$,
  as well as the corresponding induced construction of
  $\link{\kK_1}$ from the induced shelling $\shell{\link{\kK_1}}$.
  Finally, Figs. \ref{fig:K1+P1-shellings1} and
  \ref{fig:K1+P1-shellings2} show the step-by-step construction
  of $\bx{}P_1$ from the shelling $\shell{\bx{}P_1}$ induced by
  $\shell{\kK_1}$, along with the corresponding steps of the
  construction of $\kK_1$ from $\shell{\kK_1}$, \ie we only depict the
  steps of $\shell{\kK_1}$ that induce facets of $\shell{\bx{}P_1}$.

  \ShellingFigs

  Let $F$ be a facet in $\shell{\kK_1}$. If $F$ induces a facet for
  $\shell{\link{\kK_1}}$, denote by $\link{F}$ this facet of
  $\link{\kK_1}$. Similarly, if $F$ induces a facet for
  $\shell{\bx{}P_1}$, call $F_1$ this facet of $\bx{}P_1$. Finally, if
  $F_1$ induces a facet for $\shell{\link{\bx{}P_1}}$, let
  $\link{F_1}$ be this facet of $\link{\bx{}P_1}$.
  Let $G\subseteq{}F$, $\link{G}\subseteq{}\link{F}$,
  $G_1\subseteq{}F_1$ and $\link{G_1}\subseteq\link{F_1}$ be
  the minimal new faces associated with $F$, $\link{F}$, $F_1$ and
  $\link{F_1}$ in the corresponding shellings, let $\lambda$ be the
  cardinality of $G$, and observe that $F_1=F\cap\bx{}P_1$, 
  $\link{F_1}=(\link{F})\cap\bx{}P_1$, $G_1=G\cap\bx{}P_1$ and
  $\link{G_1}=(\link{G})\cap\bx{}P_1$. As long as we shell
  $\str{\kK_1}$, $G$ induces $\link{G}$, and, in fact, the faces $G$
  and $\link{G}$ coincide (see also
  Fig. \ref{fig:K1+linkv-shellings}): if $F$ is the first facet in
  $\shell{\kK_1}$, then $G\equiv\link{G}\equiv\emptyset$; otherwise,
  $v$ cannot be a vertex in $G$ or $\link{G}$ (the minimal new
  faces are faces of $\link{\kK_1}$).
  Similarly, as long as we shell $\str{\bx{}P_1}$, $G_1$ induces
  $\link{G_1}$, and, in fact, the faces $G_1$ and $\link{G_1}$
  coincide: if $F_1$ is the first facet in
  $\shell{\bx{}P_1}$, then $G_1\equiv\link{G_1}\equiv\emptyset$;
  otherwise, $v$ cannot be a vertex in $G_1$ or $\link{G_1}$ (the
  minimal new faces are faces of $\link{\bx{}P_1}$).
  Hence, as long as we shell $\str{\kK_1}$ (\ie as long as $v\in{}F$),
  we have $h_k(\link{\kK_1})=h_k(\kK_1)$ and
  $h_k(\link{\bx{P_1}})=h_k(\bx{P_1})$, for all $k\ge{}0$, and, thus,
  $h_k(\link{\kK_1})-h_k(\link{\bx{P_1}})=h_k(\kK_1)-h_k(\bx{P_1})$,
  for all $k\ge{}0$.
  After the shelling $\shell{\kK_1}$ has left $\str{\kK_1}$,
  there are no more facets in $\shell{\link{\kK_1}}$.
  This implies that, after $\shell{\kK_1}$ has left $\str{\kK_1}$ (\ie
  $v$ is not a vertex of $F$ anymore), the values of
  $h_k(\link{\kK_1})$ and $h_k(\link{\bx{P_1}})$ remain unchanged for
  all $k\ge{}0$. However,
  the values of $h_k(\kK_1)$ and $h_k(\bx{P_1})$ may increase for some
  $k$.
  More precisely, if $F$ does not induce any facet for
  $\shell{\bx{}P_1}$, then $h_\lambda(\kK_1)$ is increased by one,
  $h_k(\kK_1)$ does not change for $k\ne\lambda$, while
  $h_k(\bx{P_1})$ remains unchanged for all $k\ge{}0$.
  Thus, $h_\lambda(\link{\kK_1})-h_\lambda(\link{\bx{P_1}})<
  h_\lambda(\kK_1)-h_\lambda(\bx{P_1})$, while
  $h_k(\link{\kK_1})-h_k(\link{\bx{P_1}})\le{}h_k(\kK_1)-h_k(\bx{P_1})$,
  for all $k\ne\lambda$.
  If, however, $F$ induces $F_1$, then the minimal new face $G_1$ in
  $\shell{\bx{}P_1}$ due to $F_1$ coincides with $G$ (see also
  Figs. \ref{fig:K1+P1-shellings1} and \ref{fig:K1+P1-shellings2}).
  To vefiry this, suppose $G_1\subset{}G$; since
  $G$ is the minimal new face in $\shell{\kK_1}$, $G_1$ would have
  been a face already ``discovered'' at a previous step of
  $\shell{\kK_1}$, and thus also at a previous step of
  $\shell{\bx{}P_1}$, which contradicts the fact that $G_1$ is the
  minimal new face for $\shell{\bx{}P_1}$.
  Therefore, in this case, both $h_\lambda(\kK_1)$ and
  $h_\lambda(\bx{P_1})$ are increased by one, while $h_k(\kK_1)$
  and $h_k(\bx{P_1})$ remain unchanged for all $k\ne\lambda$.
  This implies
  $h_k(\link{\kK_1})-h_k(\link{\bx{P_1}})\le{}h_k(\kK_1)-h_k(\bx{P_1})$,
  for all $k\ge{}0$.
  Summarizing the analysis above, we deduce that for all $v\in{}V_1$,
  and for all $0\le{}k\le{}d$, we have:
  \begin{equation}\label{equ:linkbound1}
    h_k(\link{\kK_1})-h_k(\link{\bx{P_1}})\le{}h_k(\kK_1)-h_k(\bx{P_1}).
  \end{equation}
  Using the analogous argument for all vertices of $V_2$, we also get
  that, for all $v\in{}V_2$, and for all $0\le{}k\le{}d$:
  \begin{equation}\label{equ:linkbound2}
    h_k(\link{\kK_2})-h_k(\link{\bx{P_2}})\le{}h_k(\kK_2)-h_k(\bx{P_2}).
  \end{equation}
  Now, combining relation \eqref{equ:hkK} with relations
  \eqref{equ:hkKj}, yields
  \begin{align}
    h_k(\kK_1)-h_k(\bx{}P_1)=h_k(\fF)+g_k(\bx{P_2}),\label{equ:hkK1}\\
    h_k(\kK_2)-h_k(\bx{}P_2)=h_k(\fF)+g_k(\bx{P_1}),\label{equ:hkK2}
  \end{align}
  Thus, by applying relation \eqref{equ:linkbound1}, and using
  relation \eqref{equ:hkK1}, we get for every vertex $v\in{}V_1$:
  \begin{equation}
    \sum_{v\in{}V_1}[h_k(\link{\kK_1})-h_k(\link{\bx{P_1}})]
    \le\sum_{v\in{}V_1}[h_k(\kK_1)-h_k(\bx{P_1})]
    =n_1[h_k(\fF)+g_k(\bx{P_2})],
  \end{equation}
  Similarly, applying relation \eqref{equ:linkbound2}, and using
  relation \eqref{equ:hkK2}, we get for every vertex $v\in{}V_2$:
  \begin{equation}
    \sum_{v\in{}V_2}[h_k(\link{\kK_2})-h_k(\link{\bx{P_2}})]
    \le\sum_{v\in{}V_2}[h_k(\kK_2)-h_k(\bx{P_2})]
    =n_2[h_k(\fF)+g_k(\bx{P_1})].
  \end{equation}
  We thus arrive at the following inequality, for $0\le{}k\le{}d$:
  \begin{equation}\label{equ:hkFrecur0}
    (k+1)h_{k+1}(\fF)+(d+1-k)h_k(\fF)
    \le(n_1+n_2)h_k(\fF)+n_1 g_k(\bx{P_2})+n_2 g_k(\bx{P_1}),
  \end{equation}
  which gives the recurrence inequality in the statement of the lemma.
\end{proof}

Using the recurrence relation from Lemma \ref{lem:hFrecur} we get the
following bounds on the elements of $\mb{h}(\fF)$.

\begin{lemma}\label{lem:hFbound}
  For all $0\le{}k\le{}d+1$,
  \begin{equation}\label{equ:hkFbound}
    h_k(\fF)\le{}\binom{n_1+n_2-d-2+k}{k}-\binom{n_1-d-2+k}{k}
    -\binom{n_2-d-2+k}{k}.
  \end{equation}
Equality holds for all $k$ with $0\le{}k\le{}l$ if and only if
$l\le\lexp{d+1}$ and $P$ is \bn{l}.
\end{lemma}

\begin{proof}
  We show the desired bound by induction on $k$. 
  Clearly, the bound holds (as equality) for $k=0$, since
  \begin{equation}\label{equ:h0F-equality}
    h_0(\fF)=-1=1-1-1=
    \binom{n_1+n_2-d-2+0}{0}-\binom{n_1-d-2+0}{0}-\binom{n_2-d-2+0}{0}.
  \end{equation}
  Suppose now that the bound holds for $h_k(\fF)$, where
  $k\ge{}0$. 
  Using the recurrence relation \eqref{equ:hkFrecur}, in conjunction
  with the upper bounds for the elements of the $g$-vector of a
  polytope from Corollary \ref{cor:UBT},
  and since for $k\ge{}0$, $n_1+n_2-d-1+k\ge{}d+1>0$, we have
  \begin{equation}\label{equ:hkderivation}
    \begin{aligned}
      h_{k+1}(\fF)&\le\tfrac{n_1+n_2-d-1+k}{k+1}\,h_k(\fF)
      +\tfrac{n_1}{k+1}\,g_k(\bx{P_2})
      +\tfrac{n_2}{k+1}\,g_k(\bx{P_1})\\
      &\le\tfrac{n_1+n_2-d-1+k}{k+1}\,\left[\tbinom{n_1+n_2-d-2+k}{k}
        -\tbinom{n_1-d-2+k}{k}-\tbinom{n_2-d-2+k}{k}\right]\\
      &\qquad+\tfrac{n_1}{k+1}\tbinom{n_2-d-2+k}{k}
      +\tfrac{n_2}{k+1}\tbinom{n_1-d-2+k}{k}\\
      &=\tfrac{n_1+n_2-d-1+k}{k+1}\tbinom{n_1+n_2-d-2+k}{k}
      -\tfrac{n_1-d-1+k}{k+1}\tbinom{n_1-d-2+k}{k}
      -\tfrac{n_2-d-1+k}{k+1}\tbinom{n_2-d-2+k}{k}\\
      &=\tbinom{n_1+n_2-d-1+k}{k+1}-\tbinom{n_1-d-1+k}{k+1}
      -\tbinom{n_2-d-1+k}{k+1}.
    \end{aligned}
  \end{equation}

  Let us now turn to our equality claim. The claim for $l=0$ is
  obvious (cf. \eqref{equ:h0F-equality}), so we assume below that
  $l\ge{}1$.
  Suppose first that $P$ is \bn{l}. Then, we have:
  \begin{equation}\label{equ:fkFmax}
    f_{i-1}(\fF)=\binom{n_1+n_2}{i}-\binom{n_1}{i}-\binom{n_2}{i},
    \qquad 0\le{}i\le{}l.
  \end{equation}
  Substituting $f_{i-1}(\fF)$ from \eqref{equ:fkFmax} in the
  defining equations \eqref{equ:hF} for $\mb{h}(\fF)$, we get, for
  all $0\le{}k\le{}l$:
  \begin{align*}
    h_k(\fF)&=\sum_{i=0}^{k}(-1)^{k-i}\tbinom{d+1-i}{d+1-k}f_{i-1}(\fF)\\
    &=\sum_{i=0}^{k}(-1)^{k-i}\tbinom{d+1-i}{d+1-k}
    \left(\tbinom{n_1+n_2}{i}-\tbinom{n_1}{i}-\tbinom{n_2}{i}\right)\\
    &=\sum_{i=0}^{k}(-1)^{k-i}\tbinom{d+1-i}{d+1-k}\tbinom{n_1+n_2}{i}
    -\sum_{i=0}^{k}(-1)^{k-i}\tbinom{d+1-i}{d+1-k}\tbinom{n_1}{i}
    -\sum_{i=0}^{k}(-1)^{k-i}\tbinom{d+1-i}{d+1-k}\tbinom{n_2}{i}\\
    &=\tbinom{n_1+n_2-d-2+k}{k}-\tbinom{n_2-d-2+k}{k}
    -\tbinom{n_2-d-2+k}{k},
  \end{align*}
  where for the last equality we used the fact that
  $\binom{d+1-i}{d+1-k}=0$ for $i>k$, in conjunction with the
  following combinatorial identity
  (cf. \cite[eq.~(5.25)]{gkp-cm-89}, \cite[Exercise 8.20]{z-lp-95}):
  \begin{equation*}
    \sum_{0\le{}k\le{}l}\binom{l-k}{m}\binom{s}{k-n}(-1)^k
    =(-1)^{l+m}\binom{s-m-1}{l-m-n}.
  \end{equation*}
  In the equation above we set $k\sub{}i$, $l\sub{}d+1$,
  $m\sub{}d+1-k$, $n\sub{}0$, while $s$ stands for either
  $n_1+n_2$, $n_1$ or $n_2$. We thus conclude that
  \eqref{equ:hkFbound} holds as equality for all $0\le{}k\le{}l$.

  Suppose now that inequality \eqref{equ:hkFbound} holds as equality
  for all $0\le{}k\le{}l$. Substituting $h_{i}(\fF)$, $0\le{}i\le{}l$,
  from \eqref{equ:hkFbound} in \eqref{equ:facesF} we get:
  \begin{align}
    f_{l-1}(\fF)&=\sum_{i=0}^{d+1}\tbinom{d+1-i}{l-i}h_i(\fF)\notag\\
    &=\sum_{i=0}^{d+1}\tbinom{d+1-i}{l-i}
    \left(\tbinom{n_1+n_2-d-2+i}{i}
      -\tbinom{n_1-d-2+i}{i}-\tbinom{n_2-d-2+i}{i}\right)\notag\\
    &=\sum_{i=0}^{d+1}\tbinom{d+1-i}{l-i}\tbinom{n_1+n_2-d-2+i}{i}
    -\sum_{i=0}^{d+1}\tbinom{d+1-i}{l-i}\tbinom{n_1-d-2+i}{i}
    -\sum_{i=0}^{d+1}\tbinom{d+1-i}{l-i}\tbinom{n_2-d-2+i}{i}\notag\\
    &=\sum_{i=0}^{d+1}\tbinom{d+1-i}{d+1-l}\tbinom{n_1+n_2-d-2+i}{n_1+n_2-d-2}
    -\sum_{i=0}^{d+1}\tbinom{d+1-i}{d+1-l}\tbinom{n_1-d-2+i}{n_1-d-2}
    -\sum_{i=0}^{d+1}\tbinom{d+1-i}{d+1-l}\tbinom{n_2-d-2+i}{n_2-d-2}
    \label{equ:r1}\\
    &=\tbinom{(d+1)+(n_1+n_2-d-2)+1}{(d+1-l)+(n_1+n_2-d-2)+1}
    -\tbinom{(d+1)+(n_1-d-2)+1}{(d+1-l)+(n_1-d-2)+1}
    -\tbinom{(d+1)+(n_2-d-2)+1}{(d+1-l)+(n_2-d-2)+1}\label{equ:r2}\\
    &=\tbinom{n_1+n_2}{n_1+n_2-l}-\tbinom{n_1}{n_1-l}-\tbinom{n_2}{n_2-l}\notag\\
    &=\tbinom{n_1+n_2}{l}-\tbinom{n_1}{l}-\tbinom{n_2}{l},\notag
  \end{align}
  where, in order to get from \eqref{equ:r1} to \eqref{equ:r2}, we
  used the combinatorial identity (cf. \cite[eq.~(5.26)]{gkp-cm-89}):
  \begin{equation*}
    \sum_{0\le{}k\le{}l}\binom{l-k}{m}\binom{q+k}{n}=\binom{l+q+1}{m+n+1},
  \end{equation*}
  with $k\sub{}i$, $l\sub{}d+1$, $m\sub{}d+1-k$,
  $q\sub{}s-d-2$, $n\sub{}s-d-2$, and $s$ stands for either $n_1+n_2$,
  $n_1$ or $n_2$. Hence, $P$ is \bn{l}.
\end{proof}

Using the Dehn-Sommerville-like relations \eqref{equ:hFDS}, in
conjunction with the bounds from the previous lemma, we derive
alternative bounds for $h_k(\fF)$, which are of interest since they
refine the bounds for $h_k(\fF)$ from Lemma \ref{lem:hFbound} for
large values of $k$, namely for $k>\lexp{d+1}$. More precisely:

\begin{lemma}\label{lem:hFboundhigh}
  For all $0\le{}k\le{}d+1$,
  \begin{equation}\label{equ:hFboundhigh}
    h_{d+1-k}(\fF)\le\binom{n_1+n_2-d-2+k}{k}.
  \end{equation}
  Equality holds for all $k$ with $0\le{}k\le{}l$ if and only if
  $l\le\lexp{d}$ and $P$ is $l$-neighborly.
\end{lemma}

\begin{proof}
  The upper bound claim in \eqref{equ:hFboundhigh} is a direct
  consequence of the Dehn-Sommerville-like relations \eqref{equ:hFDS}
  for $\mb{h}(\fF)$, the upper bounds from Lemma
  \ref{lem:hFbound}, and the Upper Bound Theorem for polytopes as
  stated in Corollary \ref{cor:UBT}.

  The rest of the proof deals with the equality claim.
  Inequality \eqref{equ:hFboundhigh} holds as equality for all
  $0\le{}k\le{}l$, where $l\le\lexp{d}$, if and only if the following
  two conditions hold:
  \begin{enumerate}
  \item Inequalities \eqref{equ:hkFbound} hold as equalities for all
    $0\le{}k\le{}l\le\lexp{d}$.
  \item For $j=1,2$, and for all $0\le{}k\le{}l\le\lexp{d}$, we have
    $g_k(\bx{}P_j)=\binom{n_j-d-2+k}{k}$.
  \end{enumerate}
  The first condition holds true if and only if $P$ is \bn{l},
  while the second condition holds true if and only if $P_j$,
  $j=1,2$, is $l$-neighborly.
  Therefore, inequality \eqref{equ:hFboundhigh} holds as equality for
  all $0\le{}k\le{}l$ if and only if $l\le\lexp{d}$, $P$ is \bn{l}
  and both $P_1$, $P_2$ are $l$-neighborly. In view of Lemma
  \ref{lem:bn+n2n}, we conclude that equality in
  \eqref{equ:hFboundhigh} holds for all $0\le{}k\le{}l$ if and only
  if $l\le\lexp{d}$ and $P$ is $l$-neighborly.
\end{proof}

We are now ready to compute upper bounds for the face numbers of
$\fF$. Using relation \eqref{equ:facesF}, in conjunction with 
the bounds on the elements of $\mb{h}(\fF)$ from Lemma
\ref{lem:hFbound} and Lemma \ref{lem:hFboundhigh},
we get, for $0\le{}k\le{}d+1$:
\begin{align}
  f_{k-1}(\fF)&=
  \sum_{i=0}^{\lexp{d+1}}\tbinom{d+1-i}{k-i}h_i(\fF)
  +\sum_{i=\lexp{d+1}+1}^{d+1}\tbinom{d+1-i}{k-i}h_i(\fF)\notag\\
  &=\sum_{i=0}^{\lexp{d+1}}\tbinom{d+1-i}{k-i}h_i(\fF)
  +\sum_{i=0}^{\lexp{d}}\tbinom{i}{k-d-1+i}h_{d+1-i}(\fF)\notag\\
  &\le\sum_{i=0}^{\lexp{d+1}}
  \tbinom{d+1-i}{k-i}
  \Big(\tbinom{n_1+n_2-d-2+i}{i}-\sum_{j=1}^2\tbinom{n_j-d-2+i}{i}\Big)
  +\sum_{i=0}^{\lexp{d}}\tbinom{i}{k-d-1+i}
  \tbinom{n_1+n_2-d-2+i}{i}\notag\\
  &=\sum_{i=0}^{\lexp{d+1}}\tbinom{d+1-i}{k-i}\tbinom{n_1+n_2-d-2+i}{i}
  +\sum_{i=0}^{\lexp{d}}\tbinom{i}{k-d-1+i}\tbinom{n_1+n_2-d-2+i}{i}
  -\sum_{i=0}^{\lexp{d+1}}\tbinom{d+1-i}{k-i}
  \sum_{j=1}^2\tbinom{n_j-d-2+i}{i}\label{equ:e1}\\
  &=\sideset{}{^{\,*}}{\sum}_{i=0}^{\frac{d+1}{2}}
  \left(\tbinom{d+1-i}{k-i}+\tbinom{i}{k-d-1+i}\right)
  \tbinom{n_1+n_2-d-2+i}{i}
  -\sum_{i=0}^{\lexp{d+1}}\tbinom{d+1-i}{k-i}
  \sum_{j=1}^2\tbinom{n_j-d-2+i}{i}\label{equ:e2}\\
  &=f_{k-1}(C_{d+1}(n_1+n_2))-\sum_{i=0}^{\lexp{d+1}}\tbinom{d+1-i}{k-i}
  \sum_{j=1}^2\tbinom{n_j-d-2+i}{i}\notag
\end{align}
where $C_d(n)$ stands for the cyclic $d$-polytope with $n$ vertices,
$\displaystyle\sideset{}{^{\,*}}{\textstyle\sum}_{^{i=0}}^{_{\frac{\delta}{2}}}T_i$
denotes the sum of the elements $T_0, T_1,\ldots,T_{\lexp{\delta}}$
where the last term is halved if $\delta$ is even, while in order to
get from \eqref{equ:e1} to \eqref{equ:e2} we used an identity proved
in Section \ref{app:identity} of the Appendix. The following lemma
summarizes our results.

\begin{lemma}\label{lem:fkF_ub}
  For all $0\le{}k\le{}d+1$:
  \begin{equation*}
    f_{k-1}(\fF)\le{}f_{k-1}(C_{d+1}(n_1+n_2))
    -\sum_{i=0}^{\lexp{d+1}}\binom{d+1-i}{k-i}
    \left(\binom{n_1-d-2+i}{i}+\binom{n_2-d-2+i}{i}\right),
  \end{equation*}
  where $C_d(n)$ stands for the cyclic $d$-polytope with $n$ vertices.
  Furthermore:
  \begin{enumerate}
  \item
    Equality holds for all $0\le{}k\le{}l$ if and only if
    $l\le\lexp{d+1}$ and $P$ is \bn{l}.
  \item
    For $d\ge{}2$ even, equality holds for all
    $0\le{}k\le{}d+1$ if and only if $P$ is $\lexp{d}$-neighborly.
  \item
    For $d\ge{}3$ odd, equality holds for all
    $0\le{}k\le{}d+1$ if and only if $P$ is \bn{\lexp{d+1}}.
  \end{enumerate}
\end{lemma}

Since for all $1\le{}k\le{}d$, $f_{k-1}(P_1\oplus{}P_2)=f_k(\fF)$, we
arrive at the central theorem of this section, stating upper bounds
for the face numbers of the Minkowski sum of two $d$-polytopes.

\begin{theorem}\label{thm:minksum_ub}
  Let $P_1$ and $P_2$ be two $d$-polytopes in $\euc^d$, $d\ge{}2$,
  with $n_1\ge{}d+1$ and $n_2\ge{}d+1$ vertices, respectively.
  Let also $P$ be the convex hull in $\euc^{d+1}$ of
  $P_1$ and $P_2$ embedded in the hyperplanes $\{x_{d+1}=0\}$ and
  $\{x_{d+1}=1\}$ of $\euc^{d+1}$, respectively.
  Then, for $1\le{}k\le{}d$, we have:
  \begin{equation*}
    f_{k-1}(P_1\oplus{}P_2)\le{}f_k(C_{d+1}(n_1+n_2))
    -\sum_{i=0}^{\lexp{d+1}}\binom{d+1-i}{k+1-i}
    \left(\binom{n_1-d-2+i}{i}+\binom{n_2-d-2+i}{i}\right),
  \end{equation*}
  Furthermore:
  \begin{enumerate}
  \item
    Equality holds for all $1\le{}k\le{}l$ if an only if
    $l\le{}\lexp{d-1}$ and $P$ is \bn{l+1}.
  \item
    For $d\ge{}2$ even, equality holds for all $1\le{}k\le{}d$ if an
    only if $P$ is $\lexp{d}$-neighborly.
  \item
    For $d\ge{}3$ odd, equality holds for all $1\le{}k\le{}d$ if an
    only if $P$ is \bn{\lexp{d+1}}.
  \end{enumerate}
\end{theorem}

\section{Lower bounds}
\label{sec:lb}

In the previous section we proved upper bounds on the face numbers of
the Minkowski sum $P_1\oplus{}P_2$ of two polytopes $P_1$ and
$P_2$, and we provided necessary and sufficient conditions for these
bounds to hold. However, there is one remaining important question:
Are these bounds tight? In this section give a positive answer to
this question. 

We recall, from the introductory section, the already known results,
and discuss how they are related to the results in this paper.
It is already known (\eg cf. \cite{bkos-cgaa-00}) that the maximum
number of vertices/edges of the Minkowski sum of two
polygons (\ie 2-polytopes) is the sum of the vertices/edges of
the summands. These match our expressions for $d=2$ in Theorem
\ref{thm:minksum_ub}.
Fukuda and Weibel \cite{fw-fmacp-07} have shown tight expressions for
the number of $k$-faces, $0\le{}k\le{}2$, of the Minkowski sum of two
3-polytopes $P_1$ and $P_2$, as a function of the number of
vertices of $P_1$ and $P_2$. These maximal values are given in
relations \eqref{equ:ub3}, and match our expressions for $d=3$ in
Theorem \ref{thm:minksum_ub}.
In the same paper, Fukuda and Weibel have shown that given $r$
$d$-polytopes $P_1,P_2,\ldots,P_r$, the number of $k$-faces of
$P_1\oplus{}P_2\oplus{}\ldots\oplus{}P_r$ is bounded from above as per
relation \eqref{equ:ub-trivial}. These bounds have been shown to be
tight for $d\ge{}4$, $r\le\lexp{d}$, and for all $k$ with
$0\le{}k\le\lexp{d}-r$. For $r=2$, the upper bounds in
\eqref{equ:ub-trivial} reduce to 
\begin{equation}\label{equ:WeibelBound2}
  f_k(P_1\oplus{}P_2)\le
  \sum_{j=1}^{k+1}\binom{n_1}{j}\binom{n_2}{k+2-j},
  \qquad 0\le{}k\le{}d-1,
\end{equation}
and are tight for all $k$, with $0\le{}k\le{}\lexp{d}-2$.
According to Fukuda and Weibel \cite{fw-fmacp-07}, these upper bounds
are attained when considering two cyclic $d$-polytopes $P_1$ and
$P_2$, with $n_1$ and $n_2$ vertices, respectively, with disjoint
vertex sets. As we show below, this construction gives, in fact,
tight bounds on the number of $k$-faces of the Minkowski sum for all
$0\le{}k\le{}d-1$, when the dimension $d$ is even.

\begin{theorem}\label{thm:lb-even}
  Let $d\ge{}2$ and $d$ is even. Consider two cyclic $d$-polytopes $P_1$
  and $P_2$ with disjoint vertex sets on the $d$-dimensional moment
  curve, and let $n_j$ be the number of vertices of $P_j$, $j=1,2$. 
  Then, for all $1\le{}k\le{}d$:
  \begin{equation*}
    f_{k-1}(P_1\oplus{}P_2)=f_k(C_{d+1}(n_1+n_2))
    -\sum_{i=0}^{\lexp{d+1}}\binom{d+1-i}{k+1-i}
    \left(\binom{n_1-d-2+i}{i}+\binom{n_2-d-2+i}{i}\right),
  \end{equation*}
  where $C_d(n)$ stands for the cyclic $d$-polytope with $n$ vertices.
\end{theorem}

\begin{proof}
  Let $V_1$ and $V_2$ be two disjoint sets of points on the
  $d$-dimensional moment curve of cardinalities $n_1$ and $n_2$,
  respectively. Let $P_1$ and $P_2$ be the corresponding cyclic
  $d$-polytopes, and embed them, as in the previous section, in the
  hyperplanes $\{x_{d+1}=0\}$ and $\{x_{d+1}=1\}$ of $\euc^{d+1}$. Let
  $P=CH_{d+1}(\{P_1,P_2\})$ and, again as in the previous section,
  define the set of faces $\fF$ as the set of proper faces of $P$
  intersected by the hyperplane $\tilde\Pi$ with equation
  $\{x_{d+1}=\lambda\}$, $\lambda\in(0,1)$. We then get:
  \begin{equation*}
    f_{\lexp{d}-1}(\fF)=f_{\lexp{d}-2}(P_1\oplus{}P_2)
    =\sum_{j=1}^{\lexp{d}-1}\binom{n_1}{j}\binom{n_2}{\lexp{d}-j}
    =\binom{n_1+n_2}{\lexp{d}}-\binom{n_1}{\lexp{d}}-\binom{n_2}{\lexp{d}},
  \end{equation*}
  which, by Lemma \ref{lem:fbBbound}, implies that $P$ is
  \bn{\lexp{d}}.
  Using Lemma \ref{lem:bn+n2n}, in conjunction with the fact that both
  $P_1$ and $P_2$ are $\lexp{d}$-neighborly, we further conclude that
  $P$ is $\lexp{d}$-neighborly. Hence, by Theorem
  \ref{thm:minksum_ub}, our upper bounds in Theorem
  \ref{thm:minksum_ub} are attained for all face numbers of
  $P_1\oplus{}P_2$.
\end{proof}

If $d\ge{}5$ and $d$ is odd, however, the construction in
\cite{fw-fmacp-07} gives tight bounds for $f_k(P_1\oplus{}P_2)$ for all
$0\le{}k\le{}\lexp{d}-2$, which according to Theorem
\ref{thm:minksum_ub} are not sufficient to establish that the bounds
are tight for the face numbers of all dimensions. To establish the
tightness of the bounds in Theorem \ref{thm:minksum_ub} for all the
face numbers of all dimensions, we need to construct two $d$-polytopes
$P_1$ and $P_2$, with $n_1$ and $n_2$ vertices, respectively, such that
\begin{equation*}
  f_{\lexp{d}}(\fF)=f_{\lexp{d}-1}(P_1\oplus{}P_2)=
  \binom{n_1+n_2}{\lexp{d+1}}-\binom{n_1}{\lexp{d+1}}-\binom{n_2}{\lexp{d+1}},
\end{equation*}
or, equivalently, construct two $d$-polytopes $P_1$ and $P_2$, such
that $P$ is \bn{\lexp{d+1}}.

The rest of this section is devoted to this construction. Before
getting into the technical details we first outline our approach.
In what follows $d\ge{}3$ and $d$ is odd. We denote by
$\mc(t)$, $t>0$, the $(d-1)$-dimensional moment curve, \ie
$\mc(t)=(t,t^2,\ldots,t^{d-1})$, and we define two
additional curves $\mc_1(t;\zeta)$ and $\mc_2(t;\zeta)$ in
$\euc^{d+1}$, as follows:
\begin{equation}\label{equ:mc-def}
  \begin{aligned}
    \mc_1(t;\zeta)&=(t,\zeta{}t^d,t^2,t^3,\ldots,t^{d-1},0),\\
    \mc_2(t;\zeta)&=(\zeta{}t^d,t,t^2,t^3,\ldots,t^{d-1},1),
  \end{aligned}
  \qquad t>0,\qquad \zeta\ge{}0.
\end{equation}
Notice that $\mc_1(t;\zeta)$ and $\mc_2(t;\zeta)$, with $\zeta>0$, are
$d$-dimensional moment-like curves, embedded in the hyperplanes
$\{x_{d+1}=0\}$ and $\{x_{d+1}=1\}$, respectively.
%
Choose $n_1+n_2$ real numbers $\alpha_i$, $i=1,\ldots,n_1$, and
$\beta_i$, $i=1,\ldots,n_2$, such that
$0<\alpha_1<\alpha_2<\ldots<\alpha_{n_1}$ and
$0<\beta_1<\beta_2<\ldots<\beta_{n_2}$. Let $\tau$ be a strictly 
positive parameter determined below, and let $U_1$ and
$U_2$ be the $(d-1)$-dimensional point sets:
\begin{equation}\label{equ:U1U2}
  \begin{aligned}
    U_1&=\{\mc_1(\alpha_1\tau),\mc_1(\alpha_2\tau),\ldots,
    \mc_1(\alpha_{n_1}\tau)\},\\
    U_2&=\{\mc_2(\beta_1),\mc_2(\beta_2),\ldots,\mc_2(\beta_{n_2})\}.
  \end{aligned}
\end{equation}
where $\mc_j(\cdot)$ is used to denote $\mc_j(\cdot;0)$, for simplicity.
Notice that $U_1$ and $U_2$ consist of points on the moment curve
$\mc(t)$, embedded in the $(d-1)$-subspaces $\{x_1=0,x_{d+1}=0\}$ and
$\{x_2=0,x_{d+1}=1\}$ of $\euc^{d+1}$, respectively.
Call $Q_j$ the cyclic $(d-1)$-polytope defined as
the convex hull of the points in $U_j$, $j=1,2$.
We first show that, for sufficiently small $\tau$, any subset $U$
of $\lexp{d+1}$ vertices of $U_1\cup{}U_2$, such that
$U\cap{}U_j\ne\emptyset$, $j=1,2$, defines a $\lexp{d}$-face of
$Q=CH_{d+1}(\{Q_1,Q_2\})$; in other words, we show that, for
sufficiently small $\tau$, the $(d+1)$-polytope $Q$ is
\bn[U_1]{\lexp{d+1}}.
We then appropriately perturb $U_1$ and $U_2$ (by considering a
positive value for $\zeta$) so that they
become $d$-dimensional. Let $V_1$, $V_2$ be the perturbed vertex sets,
and $P_1$, $P_2$ be the resulting $d$-polytopes ($V_j$ is the vertex
set of $P_j$).
The final step of our construction amounts to considering the
$(d+1)$-polytope $P=CH_{d+1}(\{P_1,P_2\})$, and arguing that, if the 
perturbation parameter $\zeta$ is sufficiently small, then $P$ is
\bn[V_1]{\lexp{d+1}}. In view of Theorem \ref{thm:minksum_ub}, this
establishes the tightness of our bounds for all face numbers of
$P_1\oplus{}P_2$.

We start off with a technical lemma. Its proof may be found in Section
\ref{app:signdet} of the Appendix.

\newcommand{\technicallemma}[1]{
  \begin{lemma}{#1}
    Fix two integers $k\ge{}2$ and $\ell\ge{}2$, such that
    $k+\ell$ is odd. Let $D_{k,\ell}(\tau)$ be the
    $(k+\ell)\times{}(k+\ell)$ determinant:
    \begin{equation*}
      D_{k,\ell}(\tau)=
      \begin{vmatrix}
        1 &1&\cdots& 1& 0&0&\cdots& 0\\[2pt]
        x_1\tau&x_2\tau&\cdots&x_k\tau&0&0&\cdots& 0\\[2pt]
        0 &0&\cdots& 0& 1&1&\cdots& 1\\[2pt]
        0&0&\cdots&0&y_1&y_2&\cdots&y_\ell\\[2pt]
        x_1^2\tau^2&x_2^2\tau^2&\cdots&x_k^2\tau^2&
        y_1^2&y_2^2&\cdots&y_\ell^2\\[2pt]
        x_1^3\tau^3&x_2^3\tau^3&\cdots&x_k^3\tau^3&
        y_1^3&y_2^3&\cdots&y_\ell^3\\[2pt]
        \vdots&\vdots&&\vdots&\vdots&\vdots&&\vdots\\[2pt]
        x_1^m\tau^m&x_2^m\tau^m&\cdots&x_k^m\tau^m&
        y_1^m&y_2^m&\cdots&y_\ell^m
      \end{vmatrix},\qquad m=k+\ell-3,
    \end{equation*}
    where $0<x_1<x_2<\ldots<x_k$, $0<y_1<y_2<\ldots<y_\ell$, and
    $\tau>0$. Then, there exists some $\tau_0>0$ (that depends on
    the $x_i$'s, the $y_i$'s, $k$, and $\ell$) such that for all
    $\tau\in(0,\tau_0)$, the determinant $D_{k,\ell}(\tau)$ is
    strictly positive.
  \end{lemma}
}

\technicallemma{\label{lem:sign-generic-det}}

\newcounter{thmcopy}
\setcounter{thmcopy}{\thetheorem}

We now formally proceed with our construction. As described
above, consider the vertex sets $U_1$ and $U_2$
(cf. \eqref{equ:U1U2}), and call $Q_j$ the cyclic
$(d-1)$-polytope with vertex set $U_j$, $j=1,2$. Notice that
$Q_1$ (\resp $Q_2$) is embedded in the $(d-1)$-subspace
$\{x_2=0,x_{d+1}=0\}$ (\resp $\{x_1=0,x_{d+1}=1\}$) of
$\euc^{d+1}$.
As in the previous section, call $\tilde{\Pi}$ the hyperplane of
$\euc^{d+1}$ with equation $\{x_{d+1}=\lambda\}$, $\lambda\in(0,1)$.
Let $Q=CH_{d+1}(\{Q_1,Q_2\})$, and let
$\fF_{Q}$ be the set of proper faces of $Q$ with non-empty
intersection with $\tilde{\Pi}$, i.e., $\fF_{Q}$ consists of all
the proper faces of $Q$, the vertex set of which has non-empty
intersection with both $U_1$ and $U_2$.
The following lemma establishes the first step towards our
construction.

\begin{lemma}\label{lem:kKQbn}
  There exists a sufficiently small positive value $\tau^\star$ for
  $\tau$, such that the $(d+1)$-polytope $Q$ is \bn[U_1]{\lexp{d+1}}.
\end{lemma}

\begin{proof}
  Let $t_i=\alpha_i\tau$, $t_i^\epsilon=(\alpha_i+\epsilon)\tau$,
  $1\le{}i\le{}n_1$, and $s_i=\beta_i$,
  $s_i^\epsilon=\beta_i+\epsilon$, $1\le{}i\le{}n_2$, where $\epsilon>0$
  is chosen such that $\alpha_i+\epsilon<\alpha_{i+1}$, for all
  $1\le{}i<n_1$, and $\beta_i+\epsilon<\beta_{i+1}$, for all
  $1\le{}i<n_2$.

  Choose a subset $U$ of $U_1\cup{}U_2$ of size
  $\lexp{d+1}$, such that $U\cap{}U_j\ne\emptyset$, $j=1,2$.
  We denote by $\mu$ (\resp $\nu$) the cardinality of
  $U\cap{}U_1$ (\resp $U\cap{}U_2$), and, clearly,
  $\mu+\nu=\lexp{d+1}$.
  Let $\mc_1(t_{i_1}),\mc_1(t_{i_2}),\ldots,\mc_1(t_{i_\mu})$ be the
  vertices in $U\cap{}U_1$, where $i_1<i_2<\ldots<i_\mu$, and
  analogously,
  let $\mc_2(s_{j_1}),\mc_2(s_{j_2}),\ldots,\mc_2(s_{j_\nu})$ be the
  vertices in $U\cap{}U_2$, where $j_1<j_2<\ldots<j_\nu$.
  Let $\mb{x}=(x_1,x_2,\ldots,x_{d+1})$ and define the
  $(d+2)\times(d+2)$ determinant $H_U(\mb{x})$ as follows:
  \begin{equation}\label{equ:H-def}
    H_U(\mb{x})=\left|
      \begin{array}{c@{\,\,\,\,}c@{\,\,\,\,}c@{\,\,\,\,}c@{\,\,\,\,}c@{\,\,\,\,}c@{\,\,\,\,}c@{\,\,\,\,}c@{\,\,\,\,}c@{\,\,\,\,}c@{\,\,\,\,}c}
        1&1&1&\cdots&1&1&1&1&\cdots&1&1\\
        \mb{x}&\mc_1(t_{i_1})&\mc_1(t_{i_1}^\epsilon)&\cdots
        &\mc_1(t_{i_\mu})&\mc_1(t_{i_\mu}^\epsilon)
        &\mc_2(s_{j_1})&\mc_2(s_{j_1}^\epsilon)&\cdots
        &\mc_2(s_{j_\nu})&\mc_2(s_{j_\nu}^\epsilon)
      \end{array}
    \right|.
  \end{equation}
  The equation $H_U(\mb{x})=0$ is the equation of a
  hyperplane in $\euc^{d+1}$ that passes through the points in
  $U$. We claim that, for any choice of $U$,
  and for all vertices $\mb{u}$ in $(U_1\cup{}U_2)\sm{}U$,
  we have $H_U(\mb{u})>0$ for sufficiently small $\tau$.
  
  Consider first the case $\mb{u}\in{}U_1\sm{}U$. Then,
  $\mb{u}=\mc_1(t)=(t,0,t^2,t^3,\ldots,t^{d-1},0)$, $t=\alpha\tau$,
  for some $\alpha\nin\{\alpha_{i_1},\alpha_{i_2},\ldots,\alpha_{i_\mu}\}$,
  in which case $H_U(\mb{u})$ becomes:
  \begin{align*}
    H_U(\mb{u})&=\left|
      \begin{array}{c@{\,\,\,\,}c@{\,\,\,\,}c@{\,\,\,\,}c@{\,\,\,\,}c@{\,\,\,\,}c@{\,\,\,\,}c@{\,\,\,\,}c@{\,\,\,\,}c@{\,\,\,\,}c@{\,\,\,\,}c}
        1&1&1&\cdots&1&1&1&1&\cdots&1&1\\
        \mc_1(t)&\mc_1(t_{i_1})&\mc_1(t_{i_1}^\epsilon)&\cdots
        &\mc_1(t_{i_\mu})&\mc_1(t_{i_\mu}^\epsilon)
        &\mc_2(s_{j_1})&\mc_2(s_{j_1}^\epsilon)&\cdots
        &\mc_2(s_{j_\nu})&\mc_2(s_{j_\nu}^\epsilon)
      \end{array}
    \right|\\[5pt]
    &=\left|
      \begin{array}{ccccccccccc}
        1&1&1&\cdots&1&1&1&1&\cdots&1&1\\[3pt]
        t&t_{i_1}&t_{i_1}^\epsilon&\cdots&t_{i_\mu}&t_{i_\mu}^\epsilon
        &0&0&\cdots&0&0\\[3pt]
        0&0&0&\cdots&0&0&s_{j_1}&s_{j_1}^\epsilon&\cdots
        &s_{j_\nu}&s_{j_\nu}^\epsilon\\[3pt]
        t^2&t_{i_1}^2&(t_{i_1}^\epsilon)^2&\cdots&t_{i_\mu}^2&(t_{i_\mu}^\epsilon)^2
        &s_{j_1}^2&(s_{j_1}^\epsilon)^2
        &\cdots&s_{j_\nu}^2&(s_{j_\nu}^\epsilon)^2\\[3pt]
        t^3&t_{i_1}^3&(t_{i_1}^\epsilon)^3&\cdots&t_{i_\mu}^3&(t_{i_\mu}^\epsilon)^3
        &s_{j_1}^3&(s_{j_1}^\epsilon)^3
        &\cdots&s_{j_\nu}^3&(s_{j_\nu}^\epsilon)^3\\[3pt]
        \vdots&\vdots&\vdots&&\vdots&\vdots&\vdots&\vdots&&\vdots&\vdots\\[3pt]
        t^{d-1}&t_{i_1}^{d-1}&(t_{i_1}^\epsilon)^{d-1}&\cdots&t_{i_\mu}^{d-1}&
        (t_{i_\mu}^\epsilon)^{d-1}&
        s_{j_1}^{d-1}&(s_{j_1}^\epsilon)^{d-1}&\cdots&s_{j_\nu}^{d-1}
        &(s_{j_\nu}^\epsilon)^{d-1}\\[3pt]
        0&0&0&\cdots&0&0&1&1&\cdots&1&1
      \end{array}\right|.
  \end{align*}
  Observe now that we can transform $H_U(\mb{u})$ in the form
  of the determinant $D_{k,\ell}(\tau)$ of Lemma
  \ref{lem:sign-generic-det}, where $k=2\mu+1$ and $\ell=2\nu$,
  by means of the following determinant transformations:
  \begin{enumerate}
  \item Subtract the last row of $H_U(\mb{u})$ from the
    first.
  \item Shift the first column of $H_U(\mb{u})$ to the right,
    so that the non-zero values of the second row of
    $H_U(\mb{u})$ occupy columns 1 through $2\mu+1$ and are in
    increasing order. This has to be done by an \emph{even} number of
    column swaps, since $t$ cannot be between some $t_{i_k}$ and
    $t_{i_k}^\epsilon$ (due to the way we have chosen $\epsilon$).
  \item
    Shift the last row of $H_U(\mb{u})$ up, so as to become the
    third row of $H_U(\mb{u})$. This can be done by $d-1$
    row swaps, which implies that the sign of the determinant does
    not change (recall that $d$ is odd).
  \end{enumerate}

  Consider now the case $\mb{u}\in{}U_2\sm{}U$. Then,
  $\mb{u}=\mc_2(s)=(0,s,s^2,s^3,\ldots,s^{d-1},1)$,
  for some
  $s\not\in\{s_{j_1},s_{j_2},\ldots,s_{j_\nu}\}$, in
  which case $H_U(\mb{u})$ becomes:
  \begin{align*}
    H_U(\mb{u})&=\left|
      \begin{array}{c@{\,\,\,\,}c@{\,\,\,\,}c@{\,\,\,\,}c@{\,\,\,\,}c@{\,\,\,\,}c@{\,\,\,\,}c@{\,\,\,\,}c@{\,\,\,\,}c@{\,\,\,\,}c@{\,\,\,\,}c}
        1&1&1&\cdots&1&1&1&1&\cdots&1&1\\
        \mc_2(s)&\mc_1(t_{i_1})&\mc_1(t_{i_1}^\epsilon)&\cdots
        &\mc_1(t_{i_\mu})&\mc_1(t_{i_\mu}^\epsilon)
        &\mc_2(s_{j_1})&\mc_2(s_{j_1}^\epsilon)&\cdots
        &\mc_2(s_{j_\nu})&\mc_2(s_{j_\nu}^\epsilon)
      \end{array}
    \right|\\[5pt]
    &=\left|
      \begin{array}{ccccccccccc}
        1&1&1&\cdots&1&1&1&1&\cdots&1&1\\[3pt]
        0&t_{i_1}&t_{i_1}^\epsilon&\cdots&t_{i_\mu}&t_{i_\mu}^\epsilon
        &0&0&\cdots&0&0\\[3pt]
        s&0&0&\cdots&0&0&s_{j_1}&s_{j_1}^\epsilon&\cdots
        &s_{j_\nu}&s_{j_\nu}^\epsilon\\[3pt]
        s^2&t_{i_1}^2&(t_{i_1}^\epsilon)^2&\cdots
        &t_{i_\mu}^2&(t_{i_\mu}^\epsilon)^2
        &s_{j_1}^2&(s_{j_1}^\epsilon)^2&\cdots
        &s_{j_\nu}^2&(s_{j_\nu}^\epsilon)^2\\[3pt]
        s^3&t_{i_1}^3&(t_{i_1}^\epsilon)^3&\cdots
        &t_{i_\mu}^3&(t_{i_\mu}^\epsilon)^3
        &s_{j_1}^3&(s_{j_1}^\epsilon)^3&\cdots
        &s_{j_\nu}^3&(s_{j_\nu}^\epsilon)^3\\[3pt]
        \vdots&\vdots&\vdots&&\vdots&\vdots&\vdots&\vdots&&\vdots&\vdots\\[3pt]
        s^{d-1}&t_{i_1}^{d-1}&(t_{i_1}^\epsilon)^{d-1}&\cdots
        &t_{i_\mu}^{d-1}&(t_{i_\mu}^\epsilon)^{d-1}&
        s_{j_1}^{d-1}&(s_{j_1}^\epsilon)^{d-1}&\cdots
        &s_{j_\nu}^{d-1}&(s_{j_\nu}^\epsilon)^{d-1}\\[3pt]
        1&0&0&\cdots&0&0&1&1&\cdots&1&1
      \end{array}\right|.
  \end{align*}
  As for $\mb{u}\in{}U_1\sm{}U$, observe that we can
  transform $H_U(\mb{u})$ in the form of the determinant
  $D_{k,\ell}(\tau)$ of Lemma \ref{lem:sign-generic-det}, where
  now $k=2\mu$ and $\ell=2\nu+1$, by means of the following
  determinant transformations:
  \begin{enumerate}
  \item Subtract the last row of $H_U(\mb{u})$ from the
    first.
  \item Shift the first column of $H_U(\mb{u})$ to the right,
    so that the non-zero values of the third row of
    $H_U(\mb{u})$ occupy columns $2\mu+1$ through $d+2$ and are in
    increasing order. This has to be done by an \emph{even} number of
    column swaps, since we have to shift through the first
    $2\mu$ columns, and since $s$ cannot be between some $s_{j_k}$
    and $s_{j_k}^\epsilon$ (due to the way we have chosen $\epsilon$).
  \item
    Shift the last row of $H_U(\mb{u})$ up, so as to become the
    third row of $H_U(\mb{u})$. This can be done by $d-1$
    row swaps, which implies that the sign of the determinant does
    not change (recall that $d$ is odd).
  \end{enumerate}
  We thus conclude that, for any specific choice of $U$, and for
  any specific point $\mb{u}\in{}(U_1\cup{}U_2)\sm{}U$,
  there exists some $\tau_0>0$ (cf. Lemma \ref{lem:sign-generic-det})
  that depends on $\mb{u}$ and $U$, such that for all
  $\tau\in(0,\tau_0)$, $H_U(\mb{u})>0$.

  Since the total number of subsets $U$ is
  $\binom{n_1+n_2}{\lexp{d+1}}-\binom{n_1}{\lexp{d+1}}
  -\binom{n_2}{\lexp{d+1}}$, while for each such subset $U$ we need
  to consider the $(n_1+n_2-\lexp{d+1})$ vertices in
  $(U_1\cup{}U_2)\sm{}U$, it suffices to consider a value
  $\tau^\star$ for $\tau$ that is small enough, so that all
  $(n_1+n_2-\lexp{d+1})\left[\binom{n_1+n_2}{\lexp{d+1}}
    -\binom{n_1}{\lexp{d+1}}-\binom{n_2}{\lexp{d+1}}\right]$ possible
  determinants $H_U(\mb{u})$ are strictly positive. 
  Call $U_j^\star$, $j=1,2$, the vertex sets we get for
  $\tau=\tau^\star$, $Q_j^\star$ the corresponding polytopes, and
  $Q^\star$ the resulting convex hull\footnote{In fact $U_2$ is
    independent of $\tau$, but we use a unified notation for
    simplicity.}.
  Our analysis immediately implies that
  \emph{for each}
  $U^\star\subseteq{}U_1^\star\cup{}U_2^\star$, where
  $U^\star\cap{}U_j^\star\ne\emptyset$, $j=1,2$, the equation
  $H_{U^\star}(\mb{x})=0$, $\mb{x}\in\euc^{d+1}$, is the equation of a
  supporting hyperplane for $Q^\star$ passing through the 
  vertices of $U^\star$ (and those only). In other words, every
  subset of $U^\star$ of $U_1^\star\cup{}U_2^\star$, where
  $|U^\star|=\lexp{d+1}$, $U^\star\cap{}U_j^\star\ne\emptyset$, $j=1,2$,
  defines a $\lexp{d}$-face of $Q^\star$, which means that 
  $Q^\star$ is \bn[U_1^\star]{\lexp{d+1}}.
\end{proof}

We are now ready to perform the last step of our construction. 
We assume we have chosen $\tau$ to be equal to $\tau^\star$, and, as in
the proof of Lemma \ref{lem:kKQbn}, call $U_j^\star$, $Q_j^\star$,
$j=1,2$, the corresponding vertex sets and $(d-1)$-polytopes. Finally,
call $Q^\star$ the convex hull of $Q_1^\star$ and $Q_2^\star$, \ie
$Q^\star=CH_{d+1}(\{Q_1^\star,Q_2^\star\})$.
We perturb the vertex sets $U_1^\star$ and $U_2^\star$,
to get the vertex sets $V_1$ and $V_2$ by considering vertices on the
curves $\mc_1(t;\zeta)$ and $\mc_2(t;\zeta)$, with $\zeta>0$ instead
of the curves $\mc_1(t)$ and $\mc_2(t)$ (cf. \eqref{equ:mc-def}).
More precisely, define the sets $V_1$ and $V_2$ as:
\begin{equation}\label{equ:V1V2zeta}
  \begin{aligned}
    V_1&=\{\mc_1(\alpha_1\tau^\star;\zeta),\mc_1(\alpha_2\tau^\star;\zeta),
    \ldots,\mc_1(\alpha_{n_1}\tau^\star;\zeta)\},\quad\text{and}\\
    V_2&=\{\mc_2(\beta_1;\zeta),\mc_2(\beta_2;\zeta),\ldots,
    \mc_2(\beta_{n_2};\zeta)\},
  \end{aligned}
\end{equation}
where $\zeta>0$. 
Let $P_j$ be the convex hull of the vertices in $V_j$, $j=1,2$, and
notice that $P_j$ is a $\lexp{d}$-neighborly $d$-polytope.
Let $P=CH_{d+1}(\{P_1,P_2\})$, and let $\fF_P$ be the set of proper
faces of $P$ with non-empty intersection with $\tilde{\Pi}$, i.e.,
$\fF_P$ consists of all the proper faces of $P$, the vertex set of
which has non-empty intersection with both $V_1$ and $V_2$.
The following lemma establishes the final step of our construction. In
view of Theorem \ref{thm:minksum_ub}, it also establishes the
tightness of our bounds for all face numbers of $P_1\oplus{}P_2$.

\begin{lemma}
  There exists a sufficiently small positive value $\zeta^\star$ for
  $\zeta$, such that the $(d+1)$-polytope $P$ is \bn[V_1]{\lexp{d+1}}.
\end{lemma}

\begin{proof}
  Similarly to what we have done in the proof of Lemma \ref{lem:kKQbn},
  let $t_i=\alpha_i\tau^\star$, $t_i^\epsilon=(\alpha_i+\epsilon)\tau^\star$,
  $1\le{}i\le{}n_1$, and $s_i=\beta_i$,
  $s_i^\epsilon=\beta_i+\epsilon$, $1\le{}i\le{}n_2$, where $\epsilon>0$
  is chosen such that $\alpha_i+\epsilon<\alpha_{i+1}$, for all
  $1\le{}i<n_1$, and $\beta_i+\epsilon<\beta_{i+1}$, for all
  $1\le{}i<n_2$.

  Choose $V$ a subset of $V_1\cup{}V_2$ of size
  $\lexp{d+1}$, such that $V\cap{}V_j\ne\emptyset$, $j=1,2$. Denote
  by $\mu$ (\resp $\nu$) the cardinality of $V\cap{}V_1$ (\resp
  $V\cap{}V_2$).
  Considering $\zeta$ as a small positive parameter, let
  $\mc_1(t_{i_1};\zeta),\mc_1(t_{i_2};\zeta),\ldots,\mc_1(t_{i_\mu};\zeta)$
  be the vertices in $V\cap{}V_1$, where $i_1<i_2<\ldots<i_\mu$, and
  analogously,
  let
  $\mc_2(s_{j_1};\zeta),\mc_2(s_{j_2};\zeta),\ldots,\mc_2(s_{j_\nu};\zeta)$
  be the vertices in $V\cap{}V_2$, where $j_1<j_2<\ldots<j_\nu$.
  Let $\mb{x}=(x_1,x_2,\ldots,x_{d+1})$ and define the
  $(d+2)\times(d+2)$ determinant $F_V(\mb{x};\zeta)$ as:
  \begin{equation}\label{equ:Hz-def}
    F_V(\mb{x};\zeta)=
    \left|
      \begin{array}{c@{\,\,\,\,}c@{\,\,\,\,}c@{\,\,\,\,}c@{\,\,\,\,}c@{\,\,\,\,}c@{\,\,\,\,}c@{\,\,\,\,}c@{\,\,\,\,}c}
        1&1&1&\cdots&1&1&1&\cdots&1\\
        \mb{x}&\mc_1(t_{i_1};\zeta)&\mc_1(t_{i_1}^\epsilon;\zeta)&\cdots
        &\mc_1(t_{i_\mu}^\epsilon;\zeta)
        &\mc_2(s_{j_1};\zeta)&\mc_2(s_{j_1}^\epsilon;\zeta)&\cdots
        &\mc_2(s_{j_\nu}^\epsilon;\zeta)
      \end{array}
    \right|.
  \end{equation}
  The equation $F_V(\mb{x};\zeta)=0$ is the equation of a
  hyperplane in $\euc^{d+1}$ that passes through the points in
  $V$. We claim that for all vertices
  $\mb{v}\in(V_1\cup{}V_2)\sm{}V$, we have $F_V(\mb{v};\zeta)>0$ for
  sufficiently small $\zeta$.

  Indeed, let $U^\star$ denote the set of vertices in
  $U_1^\star\cup{}U_2^\star$ that correspond to the vertices in $V$, \ie
  $U^\star$ contains the projections, on the hyperplanes
  $\{x_2=0\}$ or $\{x_1=0\}$ of $\euc^{d+1}$, of the vertices in $V$,
  depending on whether these vertices belong to $V_1$ or $V_2$,
  respectively.
  Choose some $\mb{v}\in{}(V_1\cup{}V_2)\sm{}V$. If
  $\mb{v}\in{}V_1\sm{}V$, $\mb{v}$ is of the form
  $\mb{v}=\mc_1(t_i;\zeta)$, $\zeta>0$, for some
  $i\nin\{i_1,i_2,\ldots,i_\mu\}$, whereas if $\mb{v}\in{}V_2\sm{}V$,
  $\mb{v}$ is of the form $\mb{v}=\mc_2(s_j;\zeta)$, $\zeta>0$, for
  some $j\nin\{j_1,j_2,\ldots,j_\nu\}$. In the former case, let
  $\mb{u}^\star=\mc_1(t_i)=\mc_1(t_i;0)$, whereas, in the latter case, let
  $\mb{u}^\star=\mc_2(s_j)=\mc_2(s_j;0)$. In more geometric terms, we
  define $\mb{u}^\star$ to be the projection of $\mb{v}$ on the hyperplanes
  $\{x_2=0\}$ or $\{x_1=0\}$ of $\euc^{d+1}$, depending in whether
  $\mb{v}$ belongs to $V_1\sm{}V$ or $V_2\sm{}V$, respectively, or,
  equivalently, $\mb{u}^\star$ is the (unperturbed) vertex in
  $(U_1^\star\cup{}U_2^\star)\sm{}U^\star$ that corresponds to $\mb{v}$.
  Observe that $F_V(\mb{v};\zeta)$ is a polynomial function in
  $\zeta$, and thus it is continuous with respect to $\zeta$ for any
  $\zeta\in\reals$. This implies that
  \begin{equation}\label{equ:limFv}
    \lim_{\zeta\to{}0^+}F_V(\mb{v};\zeta)=F_{U^\star}(\mb{u};0)
    =H_{U^\star}(\mb{u}^\star),
  \end{equation}
  where we used the fact that $\lim_{\zeta\to{}0^+}\mb{v}=\mb{u}^\star$,
  and observed that
  $F_{U^\star}(\mb{u}^\star;0)=H_{U^\star}(\mb{u}^\star)$,
  where $H_{U^\star}(\mb{u}^\star)$ is the determinant in relation
  \eqref{equ:H-def} in the proof of Lemma \ref{lem:kKQbn}.
  Since $H_{U^\star}(\mb{u}^\star)>0$ (recall that we
  have chosen $\tau$ to be equal to $\tau^\star$), we conclude, from
  \eqref{equ:limFv}, that there exists some $\zeta_0>0$ that
  depends on $\mb{v}$ and $V$, such that for all
  $\zeta\in(0,\zeta_0)$, $F_V(\mb{v};\zeta)>0$.

  Since the total number of subsets $V$ is
  $\binom{n_1+n_2}{\lexp{d+1}}-\binom{n_1}{\lexp{d+1}}
  -\binom{n_2}{\lexp{d+1}}$, while for each such subset $V$ we need
  to consider the $(n_1+n_2-\lexp{d+1})$ vertices in
  $(V_1\cup{}V_2)\sm{}V$, it suffices to consider a value
  $\zeta^\star$ for $\zeta$ that is small enough, so that all
  $(n_1+n_2-\lexp{d+1})\left[\binom{n_1+n_2}{\lexp{d+1}}
    -\binom{n_1}{\lexp{d+1}}-\binom{n_2}{\lexp{d+1}}\right]$ possible
  determinants $F_V(\mb{v};\zeta)$ are strictly positive. 
  Call $V_j^\star$, $j=1,2$, the vertex sets we get for
  $\zeta=\zeta^\star$, $P_j^\star$ the corresponding polytopes, and
  $P^\star$ the resulting convex hull.
  Then, \emph{for each} $V^\star\subseteq{}V_1^\star\cup{}V_2^\star$, where
  $V^\star\cap{}V_j^\star\ne\emptyset$, $j=1,2$, the equation
  $F_{V^\star}(\mb{x};\zeta^\star)=0$, $\mb{x}\in\euc^{d+1}$, is the
  equation of a supporting hyperplane for $P^\star$ passing through
  the vertices of $V^\star$ (and those only). In other words, every
  subset of $V^\star$ of $V_1^\star\cup{}V_2^\star$, where
  $|V^\star|=\lexp{d+1}$, $V^\star\cap{}V_j^\star\ne\emptyset$, $j=1,2$,
  defines a $\lexp{d}$-face of $P^\star$, which means that
  $P^\star$ is \bn[V_1^\star]{\lexp{d+1}}.
\end{proof}

We are now ready to state the second main theorem of this section,
that concerns the tightness of our upper bounds on the number of
$k$-faces of the Minkowski sum of two $d$-polytopes for all
$0\le{}k\le{}d-1$ and for all odd dimensions $d\ge{}3$.

\begin{theorem}\label{thm:lb-odd}
  Let $d\ge{}3$ and $d$ is odd. There exist two $\lexp{d}$-neighborly
  $d$-polytopes $P_1$ and $P_2$ with $n_1$ and $n_2$ vertices,
  respectively, such that, for all $1\le{}k\le{}d$:
  \begin{equation*}
    f_{k-1}(P_1\oplus{}P_2)=f_k(C_{d+1}(n_1+n_2))
    -\sum_{i=0}^{\lexp{d+1}}\binom{d+1-i}{k+1-i}
    \left(\binom{n_1-d-2+i}{i}+\binom{n_2-d-2+i}{i}\right),
  \end{equation*}
  where $C_d(n)$ stands for the cyclic $d$-polytope with $n$
  vertices.
\end{theorem}
\section{Summary and  open problems}
\label{sec:concl}

In this paper we have computed the maximum number of
$k$-faces, $f_k(P_1\oplus{}P_2)$, $0\le{}k\le{}d-1$ of the Minkowski sum
of two $d$-polytopes $P_1$ and $P_2$ as a function of the
number of vertices $n_1$ and $n_2$ of these two polytopes.
In even dimensions $d\ge{}2$, these maximal values are attained if
$P_1$ and $P_2$ are cyclic $d$-polytopes with disjoint vertex sets.
In odd dimensions $d\ge{}3$, the lower bound construction is more
intricate. Denoting by $\mc_1(t;\zeta)$ and $\mc_2(t;\zeta)$ the
$d$-dimensional moment-like curves
$(t,\zeta{}t^d,t^2,t^3,\ldots,t^{d-1})$ and
$(\zeta{}t^d,t,t^2,t^3,\ldots,t^{d-1})$,
where $t>0$ and $\zeta>0$, we have shown that these maximum values are
attained if $P_1$ and $P_2$ are the $d$-polytopes with vertex sets
$V_1=\{\mc_1(\alpha_i\tau^\star;\zeta^\star)\,\,|\,\,i=1,\ldots,n_1\}$
and
$V_2=\{\mc_2(\beta_j;\zeta^\star)\,\,|\,\,j=1,\ldots,n_2\}$,
respectively, where $0<\alpha_1<\alpha_2<\ldots<\alpha_{n_1}$,
$0<\beta_1<\beta_2<\ldots<\beta_{n_2}$, and $\tau^\star$,
$\zeta^\star$ are appropriately chosen, sufficiently small, positive
parameters.

The obvious open problem is to extend our results for the Minkowski
sum of $r$ $d$-polytopes in $\euc^d$, for $r\ge{}3$ and $d\ge{}4$.
A related problem is to express the number of $k$-faces of the
Minkowski sum of $r$ $d$-polytopes in terms of the number of facets of
these polytopes. Results in this direction are known for
$d=2$ and $d=3$ only (see the introductory section and
\cite{fhw-emcms-09} for the 3-dimensional case). We would like to
derive such expressions for any $d\ge{}4$ and any number, $r$, of
summands.

\section*{Acknowledgements}
The authors would like to thank Alain Lascoux for useful
discussions related to the factorization of the determinant in Lemma
\ref{lem:sign-generic-det}, and Efi Fogel for suggestions regarding
the improvement of the presentation of the material.
The work in this paper has been partially supported by the
FP7-REGPOT-2009-1 project ``Archimedes Center for Modeling, Analysis
and Computation''.

\phantomsection
\addcontentsline{toc}{section}{References}

\bibliographystyle{plain}
\bibliography{minksum}

\phantomsection
\addcontentsline{toc}{section}{Appendix}

\appendix

\section{The summation operator}
\label{app:sumop}

Let $\yY$ be either $\fF$ or a pure simplicial subcomplex of
$\bx{}Q$. Below we compute the action of the operator
$\SSS_{k}(\cdot;\delta,\nu)$ on $\yY$, for $\nu\in\{1,2\}$ and when
$\yY$ is either $\delta$- or $(\delta-1)$-dimensional.

Recall the action of the operator $\SSS_{k}(\cdot;\delta,\nu)$ on $\yY$:
\begin{equation*}
  \SSS_{k}(\yY;\delta,\nu)
  =\sum_{i=1}^{\delta}(-1)^{k-i}\binom{\delta-i}{\delta-k}f_{i-\nu}(\yY),
\end{equation*}
and consider first the case where $\yY$ is $\delta$-dimensional and
$\nu=1$. In this case we have:
\begin{equation}\label{equ:sum1}
  \begin{aligned}
    \SSS_{k}(\yY;\delta,1)
    &=\sum_{i=1}^{\delta}(-1)^{k-i}\binom{\delta-i}{\delta-k}f_{i-1}(\yY)\\
    &=\sum_{i=0}^{\delta}(-1)^{k-i}\binom{\delta-i}{\delta-k}f_{i-1}(\yY)
    -(-1)^{k}\binom{\delta}{\delta-k}f_{-1}(\yY)\\
    &=h_k(\yY)-(-1)^k\binom{\delta}{\delta-k}f_{-1}(\yY).
  \end{aligned}
\end{equation}
If $\yY$ is $(\delta-1)$-dimensional and $\nu=1$, we have:
\begin{equation}\label{equ:sum2}
  \begin{aligned}
    \SSS_{k}(\yY;\delta,1)&=\sum_{i=1}^{\delta}(-1)^{k-i}
    \binom{\delta-i}{\delta-k}f_{i-1}(\yY)\\
    &=\sum_{i=1}^{\delta}(-1)^{k-i}\left(\binom{\delta-i-1}{\delta-k}
      +\binom{\delta-i-1}{\delta-k-1}\right)f_{i-1}(\yY)\\
    &=-\sum_{i=1}^{\delta}(-1)^{(k-1)-i}
    \binom{\delta-1-i}{\delta-1-(k-1)}f_{i-1}(\yY)
    +\sum_{i=1}^{\delta}(-1)^{k-i}\binom{\delta-1-i}{\delta-1-k}f_{i-1}(\yY)\\
    &=-h_{k-1}(\yY)-(-1)^k\binom{\delta-1}{\delta-k}f_{-1}(\yY)
    +h_k(\yY)-(-1)^k\binom{\delta-1}{\delta-1-k}f_{-1}(\yY)\\
    &=h_k(\yY)-h_{k-1}(\yY)-(-1)^k\binom{\delta}{\delta-k}f_{-1}(\yY).
  \end{aligned}
\end{equation}
Finally, if $\yY$ is $(\delta-1)$-dimensional and $\nu=2$, we have:
\begin{equation}\label{equ:sum3}
  \begin{aligned}
    \SSS_{k}(\yY;\delta,2)&=
    \sum_{i=1}^{\delta}(-1)^{k-i}\binom{\delta-i}{\delta-k}f_{i-2}(\yY)\\
    &=\sum_{i=0}^{\delta-1}(-1)^{(k-1)-i}\binom{\delta-1-i}{\delta-1-(k-1)}
    f_{i-1}(\yY)\\
    &=h_{k-1}(\yY)
  \end{aligned}
\end{equation}


\section{Proof of an identity}
\label{app:identity}

In this section we prove the following identity used in
Section \ref{sec:ub} to prove the upper bound for $f_{k-1}(\fF)$
(see relations \eqref{equ:e1} and \eqref{equ:e2}).

\begin{lemma}
For any $d\ge{}2$, and any sequence of numbers $\alpha_i$, where
$0\le{}i\le{}\lexp{d+1}$, we have:
\begin{equation*}
  \sum_{i=0}^{\lexp{d+1}}\binom{d+1-i}{k-i}\alpha_i
  +\sum_{i=0}^{\lexp{d}}\binom{i}{k-d-1+i}\alpha_i
  =\sideset{}{^{\,*}}{\sum}_{i=0}^{\frac{d+1}{2}}
  \left(\binom{d+1-i}{k-i}+\binom{i}{k-d-1+i}\right)\alpha_i.
\end{equation*}
\end{lemma}

\begin{proof}
  We start by recalling the definition of the symbol
  $\displaystyle\sideset{}{^{\,*}}{\sum}_{i=0}^{\frac{\delta}{2}}T_i$. This
  symbol denotes the sum of the elements
  $T_0,T_1,\ldots,T_{\lexp{\delta}}$, where the last term is halved if
  $\delta$ is even. More precisely:
  \begin{equation*}
    \sideset{}{^{\,*}}{\sum}_{i=0}^{\frac{\delta}{2}}T_i=
    \begin{cases}
      T_0+T_1+\ldots+T_{\lexp{\delta}-1}+\frac{1}{2}T_{\lexp{\delta}}
      &\text{if $\delta$ is even},\\
      T_0+T_1+\ldots+T_{\lexp{\delta}-1}+T_{\lexp{\delta}}
      &\text{if $\delta$ is odd}.
    \end{cases}
  \end{equation*}

  \bigskip\noindent
  Let us now first consider the case $d$ odd. In this case $d+1$ is
  even, and we have:
  \begin{equation*}
    \begin{aligned}
      \sum_{i=0}^{\lexp{d+1}}\tbinom{d+1-i}{k-i}\alpha_i
      &+\sum_{i=0}^{\lexp{d}}\tbinom{i}{k-d-1+i}\alpha_i
      =\sum_{i=0}^{\lexp{d+1}}\tbinom{d+1-i}{k-i}\alpha_i
      +\sum_{i=0}^{\lexp{d+1}-1}\tbinom{i}{k-d-1+i}\alpha_i\\
      &=\sum_{i=0}^{\lexp{d+1}-1}\left(\tbinom{d+1-i}{k-i}
        +\tbinom{i}{k-d-1+i}\right)\alpha_i
      +\tbinom{d+1-\lexp{d+1}}{k-\lexp{d+1}}\alpha_{\lexp{d+1}}\\
      &=\sum_{i=0}^{\lexp{d+1}-1}\left(\tbinom{d+1-i}{k-i}
        +\tbinom{i}{k-d-1+i}\right)\alpha_i
      +{\textstyle\frac{1}{2}}\left(\tbinom{d+1-\lexp{d+1}}{k-\lexp{d+1}}
        +\tbinom{\lexp{d+1}}{k-d-1+\lexp{d+1}}\right)\alpha_{\lexp{d+1}}\\
      &=\sideset{}{^{\,*}}{\sum}_{i=0}^{\frac{d+1}{2}}
      \left(\tbinom{d+1-i}{k-i}+\tbinom{i}{k-d-1+i}\right)\alpha_i
    \end{aligned}
  \end{equation*}
  The case $d$ even is even simpler to prove. In this case $d+1$ is
  odd, hence:
  \begin{equation*}
    \begin{aligned}
      \sum_{i=0}^{\lexp{d+1}}\tbinom{d+1-i}{k-i}\alpha_i
      +\sum_{i=0}^{\lexp{d}}\tbinom{i}{k-d-1+i}\alpha_i
      &=\sum_{i=0}^{\lexp{d+1}}\tbinom{d+1-i}{k-i}\alpha_i
      +\sum_{i=0}^{\lexp{d+1}}\tbinom{i}{k-d-1+i}\alpha_i\\
      &=\sum_{i=0}^{\lexp{d+1}}\left(\tbinom{d+1-i}{k-i}
        +\tbinom{i}{k-d-1+i}\right)\alpha_i\\
      &=\sideset{}{^{\,*}}{\sum}_{i=0}^{\frac{d+1}{2}}
      \left(\tbinom{d+1-i}{k-i}+\tbinom{i}{k-d-1+i}\right)\alpha_i
    \end{aligned}
  \end{equation*}
  This completes the proof.
\end{proof}

\section{Proof of Lemma 15}
\label{app:signdet}

\newcommand{\compl}{\bar{S}}
\newcommand{\cb}[1]{\mb{\bar{#1}}}

We start by introducing what is known as
\emph{Laplace's Expansion Theorem} for determinants
(see \cite{g-tm-60,hk-la-71} for details and proofs).
Consider a $n\times{}n$ matrix $A$.
Let $\mb{r}=(r_1,r_2,\ldots,r_k)$, be a vector of $k$ row indices for
$A$, where $1\le{}k<n$ and $1\le{}r_1<r_2<\ldots<r_k\le{}n$.
Let $\mb{c}=(c_1,c_2,\ldots,c_k)$ be a vector of $k$ column indices for
$A$, where $1\le{}k<n$ and $1\le{}c_1<c_2<\ldots<c_k\le{}n$. We
denote by $S(A;\mb{r},\mb{c})$ the $k\times{}k$ submatrix of $A$
constructed by keeping the entries of $A$ that belong to a row in
$\mb{r}$ and a column in $\mb{c}$.
The \emph{complementary submatrix} for $S(A;\mb{r},\mb{c})$, denoted
by $\compl(A;\mb{r},\mb{c})$, is the $(n-k)\times(n-k)$ submatrix of
$A$ constructed by removing the rows and columns of $A$ in $\mb{r}$
and $\mb{c}$, respectively. Then, the determinant of $A$ can be
computed by expanding in terms of the $k$ columns of $A$ in
$\mb{c}$ according to the following theorem.

\begin{theorem}[\textbf{Laplace's Expansion Theorem}]\label{thm:LET}
  Let $A$ be a $n\times{}n$ matrix. Let
  $\mb{c}=(c_1,c_2,\ldots,c_k)$ be a vector of $k$ column indices for
  $A$, where $1\le{}k<n$ and $1\le{}c_1<c_2<\ldots<c_k\le{}n$. Then:
  \begin{equation}
    \det(A)=\sum_{\mb{r}}(-1)^{|\mb{r}|+|\mb{c}|}
    \det(S(A;\mb{r},\mb{c}))\det(\compl(A;\mb{r},\mb{c})),
  \end{equation}
  where $|\mb{r}|=r_1+r_2+\ldots+r_k$,
  $|\mb{c}|=c_1+c_2+\ldots+c_k$, and the summation is taken over all
  row vectors $\mb{r}=(r_1,r_2,\ldots,r_k)$ of $k$ row indices for
  $A$, where $1\le{}r_1<r_2<\ldots<r_k\le{}n$.
\end{theorem}

\newcommand{\VD}[2][]{\text{VD}(#1\mb{#2})}
\newcommand{\GVD}{\text{GVD}}

The next item that will be useful is some notation and discussion
about Vandermonde and generalized Vandermonde determinants. Given a
vector of $n\ge{}2$ real numbers $\mb{x}=(x_1,x_2,\ldots,x_n)$, the
\emph{Vandermonde determinant} $\VD{x}$ of $\mb{x}$ is the
$n\times{}n$ determinant
\begin{equation*}
  \VD{x}=\left|
    \begin{array}{ccccc}
      1&1&\cdots&1\\
      x_1&x_2&\cdots&x_n\\
      x_1^2&x_2^2&\cdots&x_n^2\\
      \vdots&\vdots&&\vdots\\
      x_1^{n-1}&x_2^{n-1}&\cdots&x_n^{n-1}\\
    \end{array}
  \right|=\prod_{1\le{}i<j\le{}n}(x_j-x_i)
\end{equation*}
From the above expression, it is readily seen that if the elements of
$\mb{x}$ are in strictly increasing order, then $\VD{x}>0$. A
generalization of the Vandermonde determinant is the generalized
Vandermonde determinant: if, in addition to $\mb{x}$, we specify a vector
of exponents $\mb{\mu}=(\mu_1,\mu_2,\ldots,\mu_n)$, where
we require that $0\le{}\mu_1<\mu_2<\ldots<\mu_n$, we can
define the \emph{generalized Vandermonde determinant}
$\GVD(\mb{x};\mb{\mu})$ as the $n\times{}n$ determinant:
\begin{equation*}
  \GVD(\mb{x};\mb{\mu})=\left|
    \begin{array}{ccccc}
      x_1^{\mu_1}&x_2^{\mu_1}&\cdots&x_n^{\mu_1}\\
      x_1^{\mu_2}&x_2^{\mu_2}&\cdots&x_n^{\mu_2}\\
      x_1^{\mu_3}&x_2^{\mu_3}&\cdots&x_n^{\mu_3}\\
      \vdots&\vdots&&\vdots\\
      x_1^{\mu_n}&x_2^{\mu_n}&\cdots&x_n^{\mu_n}\\
    \end{array}
  \right|.
\end{equation*}
It is a well-known fact that, if the elements of $\mb{x}$ are in
strictly increasing order, then $\GVD(\mb{x};\mb{\mu})>0$ (for
example, see \cite{g-atm-05} for a proof of this fact).

Before proceeding with the proof of Lemma \ref{lem:sign-generic-det}
we need to introduce some additional notation concerning vectors.
We denote by $\mb{e}_i$ the vector whose elements are zero except for
the $i$-th element, which is equal to 1.
Given two vectors of size $n$ $\mb{a}=(a_1,a_2,\ldots,a_n)$ and
$\mb{b}=(b_1,b_2,\ldots,b_n)$, we denote by $\mb{a}-\mb{b}$ the
vector we get by element-wise subtracting the elements of the second
vector from the elements of the first, i.e.,
$\mb{a}-\mb{b}=(a_1-b_1,a_2-b_2,\ldots,a_n-b_n)$. Finally, given
some $t\in\reals$, and a vector $\mb{x}=(x_1,x_2,\ldots,x_n)$, we
denote by $t\mb{x}$ 
the vector $(tx_1,tx_2,\ldots,tx_n)$.

We now restate Lemma \ref{lem:sign-generic-det} and prove it.

\newcounter{thmcopytmp}
\setcounter{thmcopytmp}{\thetheorem}
\setcounter{theorem}{\thethmcopy}
\addtocounter{theorem}{-1}

\technicallemma{}

\newcommand{\md}{\Delta_{k,\ell}(\tau)}
\newcommand{\mddet}{D_{k,\ell}(\tau)}

\begin{proof}
  We denote by $\md$ the matrix corresponding to the determinant
  $D_{k,\ell}(\tau)$. We are now going to apply Laplace's Expansion
  Theorem to evaluate $D_{k,\ell}(\tau)$ in terms of the first $k$
  columns of $\md$. Note that in this case $\mb{c}=(1,2,\ldots,k)$,
  so we get:
  \begin{equation}\label{equ:DLET}
    \begin{aligned}
      \mddet&=\sum_{\mb{r}}(-1)^{|\mb{r}|+|\mb{c}|}
      \det(S(\md;\mb{r},\mb{c}))\det(\compl(\md;\mb{r},\mb{c}))\\
      &=(-1)^{\frac{k(k+1)}{2}}\sum_{\mb{r}}(-1)^{|\mb{r}|}
      \det(S(\md;\mb{r},\mb{c}))\det(\compl(\md;\mb{r},\mb{c})).
    \end{aligned}
  \end{equation}
  It is easy to verify that the above sum consists of
  $\binom{k+\ell}{k}$ terms. Observe that, among these terms:
  \begin{enumerate}
  \item
    all terms for which $\mb{r}$ contains the third or the fourth row
    vanish (the corresponding row of $S(\md;\mb{r},\mb{c})$ consists
    of zeros), and
  \item
    all terms for which $\mb{r}$ does not contain the first or the
    second row vanish (in this case there exists at least one row of
    $\compl(\md;\mb{r},\mb{c})$ that consists of zeros).
  \end{enumerate}
  The remaining terms of the expansion are the $\binom{k+\ell-4}{k-2}$
  terms for which $\mb{r}=(1,2,r_3,r_4,\ldots,r_k)$, with
  $5\le{}r_3<r_4<\ldots<r_k\le{}k+\ell$. For any given such
  $\mb{r}$, we have that:
  \begin{enumerate}
  \item
    $\det(S(\md,\mb{r},\mb{c}))$ is the $k\times{}k$
    generalized Vandermonde determinant
    $\GVD(\tau\mb{x};\mb{r}-\mb{\alpha})$, where
    $\tau\mb{x}=(\tau{}x_1,\tau{}x_2,\ldots,\tau{}x_k)$,
    $\mb{\alpha}=(1,1,3,3,\ldots,3)=\mb{e}_1+\mb{e}_2+3\sum_{i=3}^k\mb{e}_i$,
    and
  \item
    $\det(\compl(\md,\mb{r},\mb{c}))$ is the $\ell\times{}\ell$
    generalized Vandermonde determinant
    $\GVD(\mb{y};\cb{r}-\mb{\beta})$, where $\cb{r}$ is the vector
    of the $\ell$, among the $k+\ell$, row indices for $\md$ that do
    not belong to $\mb{r}$, and
    $\mb{\beta}=(3,3,\ldots,3)=3\sum_{i=1}^\ell\mb{e}_i$.
  \end{enumerate}
  We can, thus, simplify the expansion in \eqref{equ:DLET} to get:
  \begin{equation}\label{equ:DLET-simple}
    \mddet=(-1)^{\frac{k(k+1)}{2}}
    \sum_{\substack{\mb{r}=(1,2,r_3,\ldots,r_k)\\5\le{}r_3<r_4<\ldots<r_k\le{}k+\ell}}
    (-1)^{|\mb{r}|}\,\GVD(\tau\mb{x};\mb{r}-\mb{\alpha})\,
    \GVD(\mb{y};\cb{r}-\mb{\beta})
  \end{equation}
  Notice that $\GVD(\tau\mb{x};\mb{r}-\mb{\alpha})
  =\tau^{|\mb{r}-\mb{\alpha}|}\,\GVD(\mb{x};\mb{r}-\mb{\alpha})
  =\tau^{|\mb{r}|-|\mb{\alpha}|}\,\GVD(\mb{x};\mb{r}-\mb{\alpha})$.
  This means that the minimum exponent for $\tau$ is attained when
  $|\mb{r}|$ is minimal, which is the case when $\mb{r}$ is equal to
  $\mb{\rho}=(1,2,5,6,\ldots,k+2)$. For this value for $\mb{r}$, we
  also have that $\GVD(\mb{x};\mb{\rho}-\mb{\alpha})=\VD{x}$,
  while $\cb{r}$ is equal to $\cb{\rho}=(3,4,k+3,k+4,\ldots,k+\ell)$.
  Hence we get:
  \begin{equation}\label{equ:DLET-asymp}
    \mddet=(-1)^{\frac{k(k+1)}{2}+|\mb{\rho}|}\,\tau^{|\mb{\rho}|-|\mb{\alpha}|}
    \,\VD{x}\,\GVD(\mb{y};\cb{\rho}-\mb{\beta})
    +O(\tau^{|\mb{\rho}|-|\mb{\alpha}|+1}).
  \end{equation}
  Since
  $|\mb{\rho}|=\sum_{i=1}^{k+2}i-(3+4)=\sum_{i=1}^ki+(k+1)+(k+2)-7=
  \frac{k(k+1)}{2}+2k-4$, while $|\mb{\alpha}|=1+1+3(k-2)=3k-4$,
  relation \eqref{equ:DLET-asymp} can be rewritten as:
  \begin{equation}\label{equ:DLET-asymp2}
    \mddet=\tau^{\frac{k(k-1)}{2}}
    \,\VD{x}\,\GVD(\mb{y};\cb{\rho}-\mb{\beta})
    +O(\tau^{\frac{k(k-1)}{2}+1}),
  \end{equation}
  where we also used the fact that
  $(-1)^{\frac{k(k+1)}{2}+|\mb{\rho}|}=(-1)^{k(k+1)+2k-4}=1$, since
  $k(k+1)$ is even for all $k$. From relation
  \eqref{equ:DLET-asymp2} we immediately deduce that:
  \begin{equation}
    \lim_{\tau\to{}0^+}\frac{\mddet}{\tau^{\frac{k(k-1)}{2}}}
    =\VD{x}\,\GVD(\mb{y};\cb{\rho}-\mb{\beta}),
  \end{equation}
  which establishes the claim of the lemma, since both $\VD{x}$ and
  $\GVD(\mb{y};\cb{\rho}-\mb{\beta})$ are strictly positive.

  We end the proof of the lemma by commenting on two special cases:
  $k=2$ and $\ell=2$. In these two cases there is a single
  non-vanishing term in the expansion of $\mddet$, namely, the term
  corresponding to $\mb{r}=(1,2)$, if $k=2$, and
  $\mb{r}=(1,2,5,6,\ldots,k+2)$, if $\ell=2$.
  More precisely, if $k=2$, then
  $\mb{\alpha}=(1,1)=\mb{e}_1+\mb{e}_2$,
  $\cb{r}=(3,4,5,6,\ldots,\ell+2)$, and, thus
  \begin{align*}
    \det(S(\md,\mb{r},\mb{c}))&
    =\GVD(\tau\mb{x};\mb{r}-\mb{\alpha})=\VD[\tau]{x}=\tau\,\VD{x}
    =\tau(x_2-x_1),\qquad\text{and}\\
    \det(\compl(\md,\mb{r},\mb{c}))&
    =\GVD(\mb{y};\cb{r}-\mb{\beta})=\VD{y}.
  \end{align*}
  If $\ell=2$, then $\cb{r}=(3,4)$, and, thus,
  \begin{align*}
    \det(S(\md,\mb{r},\mb{c}))&
    =\GVD(\tau\mb{x};\mb{r}-\mb{\alpha})=\VD[\tau]{x}
    =\tau^{\frac{k(k-1)}{2}}\,\VD{x},\qquad\text{and}\\
    \det(\compl(\md,\mb{r},\mb{c}))&
    =\GVD(\mb{y};\cb{r}-\mb{\beta})=\VD{y}=y_2-y_1.
  \end{align*}
  Hence, in both cases, we have:
  \begin{equation*}
    \mddet=(-1)^{\frac{k(k+1)}{2}+|\mb{r}|}\,
    \GVD(\tau\mb{x};\mb{r}-\mb{\alpha})\,
    \GVD(\mb{y};\cb{r}-\mb{\beta})
    =\tau^{\frac{k(k-1)}{2}}\,\VD{x}\,\VD{y},
  \end{equation*}
  which is strictly positive for any $\tau>0$.
\end{proof}

\setcounter{theorem}{\thethmcopytmp}

\end{document}